\documentclass[11pt,draftcls,onecolumn]{IEEEtran}

\usepackage{graphicx}
\usepackage{pb-diagram}
\usepackage[centertags]{amsmath}
\usepackage{amsfonts}
\usepackage{amssymb}

\usepackage{standalone}
\usepackage{tikz}
\usepackage{xparse}
\usepackage{verbatim}
\usetikzlibrary{arrows}
\usetikzlibrary{calc,matrix,fit,shapes,patterns,backgrounds}

\usepackage[ruled,vlined]{algorithm2e}

\usepackage{cite}

\usepackage[T1]{fontenc}

\usepackage{url}
\usepackage{booktabs}
\usepackage[section]{placeins}
\usepackage{todonotes}

\newcommand{\CHI}{\mathcal{H}_{I}}

\def\v{\dot{v}}

\def\C{{\mathbb C}}

\def\R{{\mathbb R}}

\def \tH {\tilde{H}}

\def\S{\mathbb{S}}
\def\CS{\mathcal{S}}

\def\CV{\mathcal{V}}

\def\CU{\mathcal{U}}
\def\CH{\mathcal{H}}

\def\ra{{\rightarrow}}

\newcommand{\G}[2]{\mathbb{G}_{#1,#2}}

\DeclareMathOperator*{\per}{per}
\DeclareMathOperator*{\rank}{rank}

\def\Xint#1{\mathchoice
{\XXint\displaystyle\textstyle{#1}}%
{\XXint\textstyle\scriptstyle{#1}}%
{\XXint\scriptstyle\scriptscriptstyle{#1}}%
{\XXint\scriptscriptstyle\scriptscriptstyle{#1}}%
\!\int}
\def\XXint#1#2#3{{\setbox0=\hbox{$#1{#2#3}{\int}$}
\vcenter{\hbox{$#2#3$}}\kern-.5\wd0}}
\def\dashint{\Xint-}

\newcommand{\bi}{\begin{itemize}}
\newcommand{\ei}{\end{itemize}}
\newcommand{\bd}{\begin{description}}
\newcommand{\ed}{\end{description}}
\newcommand{\beq}{\begin{equation}}
\newcommand{\eeq}{\end{equation}}
\newcommand{\beqn}{\begin{eqnarray}}
\newcommand{\eeqn}{\end{eqnarray}}
\newcommand{\beqna}{\begin{eqnarray*}}
\newcommand{\eeqna}{\end{eqnarray*}}

\newenvironment{DIFnomarkup}{}{}
\newtheorem{corollary}{Corollary}
\newtheorem{lemma}{Lemma}
\newtheorem{theorem}{Theorem}
\newtheorem{proposition}{Proposition}

\newtheorem{definition}{Definition}[section]

\newtheorem{example}{Example}

\newtheorem{remark}{Remark}

\makeatletter
\newcommand{\ignorecite}[1]{{\@fileswfalse\cite{#1}}}%
\makeatother

\begin{document}
%
\title{On the Number of Interference Alignment Solutions for the K-User MIMO Channel with Constant Coefficients}
%
%
\author{\'Oscar~Gonz\'alez, \IEEEmembership{Student~Member,~IEEE}, Carlos~Beltr\'an, and~Ignacio~Santamar\'ia,~\IEEEmembership{Senior~Member,~IEEE}\thanks{\'O. Gonz\'alez and I. Santamar\'ia are with the Communications Engineering Department (DICOM), University of Cantabria, Santander, 39005, Spain. C. Beltr\'an is with the Departamento de Matem\'aticas, Estad\'istica y Computaci\'on, Universidad de Cantabria. Avda. Los Castros s/n, Santander, Spain. The work of \'O. Gonz\'alez and I. Santamar\'ia was supported by MICINN (Spanish Ministry for Science and Innovation) under grants TEC2010-19545-C04-03 (COSIMA), CONSOLIDER-INGENIO 2010 CSD2008-00010 (COMONSENS) and FPU grant AP2009-1105. Carlos Beltr\'an was partially supported by MICINN grant MTM2010-16051.

This paper was presented in part at the IEEE 2013 International Symposium on Information Theory (ISIT 2013), Istanbul, Turkey \protect\ignorecite{Gonzalez2013}.
}}%

\maketitle

\begin{abstract}
In this paper, we study the number of different interference alignment (IA) solutions in a $K$-user multiple-input multiple-output (MIMO) interference channel, when the alignment is performed via beamforming and no symbol extensions are allowed. We focus on the case where the number of IA equations matches the number of variables. In this situation, the number of IA solutions is finite and constant for any channel realization out of a zero-measure set and, as we prove in the paper, it is given by an integral formula that can be numerically approximated using Monte Carlo integration methods. More precisely, the number of alignment solutions is the scaled average of the determinant of a certain Hermitian matrix related to the geometry of the problem. Interestingly, while the value of this determinant at an arbitrary point can be used to check the feasibility of the IA problem, its average (properly scaled) gives the number of solutions. For single-beam systems the asymptotic growth rate of the number of solutions is analyzed and some connections with classical combinatorial problems are presented. Nonetheless, our results can be applied to arbitrary interference MIMO networks, with any number of users, antennas and streams per user.
\end{abstract}
\clearpage

\begin{center}
\textbf{Editorial Area}\\
Communications
\end{center}



\section{Introduction}

Interference alignment (IA) has received a lot of attention in recent years as a key technique to achieve the maximum degrees of freedom (DoF) of wireless networks in the presence of interference. Originally proposed in \cite{Jafar08},\cite{Motahari08}, the basic idea of IA consists of designing the transmitted signals in such a way that the interference at each receiver falls within a lower-dimensional subspace, therefore leaving a subspace free of interference for the desired signal \cite{JafarTut}. This idea has been applied in different forms (e.g., ergodic interference alignment \cite{Nazer12}, signal space alignment \cite{Jafar08}, or signal scale alignment \cite{Bresler2010},\cite{Cadambe2009c}), and adapted to various wireless networks such as interference networks \cite{Jafar08}, X channels \cite{Motahari08}, downlink broadcast channels in cellular communications \cite{TseDownlink11} and, more recently, to two-hop relay-aided networks in the form of interference neutralization \cite{Gou12}.

In this paper we consider the linear IA problem (i.e., signal space alignment by means of linear beamforming) for the $K$-user multiple-input multiple-output (MIMO) interference channel with constant channel coefficients. Moreover, the MIMO channels are considered to be generic, without any particular structure, which happens, for instance, when the channel matrices have independent entries drawn from a continuous distribution. This setup has also been the preferred option for recent experimental studies on IA \cite{Katabi09},\cite{Ayach10},\cite{Oscar11}.

The feasibility of linear IA for MIMO interference networks, which amounts to study the solvability of a set of polynomial equations, has been an active research topic during the last years \cite{Yetis10},\cite{Razaviyayn1},\cite{BreslerCartwrightTseToappear},\cite{Slock10},\cite{LauISIT12}. Combining algebraic geometry tools with differential topology, it has been recently proved in \cite{Gonzalez2012b} that an IA problem with any number of users, antennas and streams per user, is feasible \textit{iff} the linear mapping given by the projection from the tangent space of $\CV$ (the solution variety, whose elements are the triplets formed by the channels, decoders and precoders satisfying the IA equations) to the tangent space of $\mathcal{H}$ (the complex space of MIMO interference channels) at some element of $\CV$ is surjective. Note that this implies, in particular, that the dimension of $\CV$ must be larger than or equal to the dimension of $\mathcal{H}$ \cite{Razaviyayn1},\cite{BreslerCartwrightTseToappear}.

Exploiting this result, a general IA feasibility test with polynomial complexity has also been proposed in \cite{Gonzalez2012b}. This test reduces to check whether the determinant of a given square Hermitian matrix is zero (meaning infeasible almost surely) or not (feasible).

In this paper we build on the results in \cite{Gonzalez2012b} to study the problem of how many different alignment solutions exist for a given IA scenario. While the number of solutions is known for some particular cases (e.g the 3-user interference channel \cite{Bresler2011_3user}), a general result is not available yet. In \cite{Gonzalez2012b} it was proved that systems for which the algebraic dimension of the solution variety is strictly larger than that of the input space can have either zero or an infinite number of alignment solutions. In plain words, these are MIMO interference networks for which the number of variables is larger than the number of equations in the polynomial system. On the other hand, systems with less variables than equations are always infeasible \cite{Razaviyayn1,BreslerCartwrightTseToappear,Gonzalez2012b}.
Herein we will focus on the case in between, where the dimensions of $\CV$ and $\mathcal{H}$ are exactly the same (identical number of variables and equations), and consequently, the number of IA solutions is finite (it may be even zero) and constant out of a zero measure set of $\mathcal{H}$ as also proved in \cite{Gonzalez2012b}. In summary, rather than just characterizing feasible or infeasible system configurations, we seek to provide a more refined answer to the feasibility problem.

The number of solutions for single-beam  MIMO networks (i.e., all users wish to transmit a single stream of data) follows directly from a classical result from algebraic geometry, Bernstein's Theorem, as shown in \cite{Yetis10}.
More specifically, the number of alignment solutions coincides with the mixed volume of the Newton polytopes that support each equation of the polynomial system. Although this solves theoretically the problem for single-beam networks, in practice the computation of the mixed volume of a set of IA equations using the available software tools \cite{Lee07mixedvolume} can be very demanding. As a consequence, only a few cases have been solved so far. For single-beam networks, some upper bounds on the number of solutions using Bezout's Theorem have also been proposed in\cite{Yetis10},\cite{Schmidt2010}. For multi-beam scenarios, however, the genericity of the polynomials system of equations is lost and it is not possible to resort to mixed volume calculations to find the number of solutions. Furthermore, the existing bounds in multi-beam cases are very loose.

The main contribution of this paper is an integral formula for the number of IA solutions for arbitrary feasible networks. More specifically, we prove that while the feasibility problem is solved by checking the determinant of a certain Hermitian matrix, the number of IA solutions is given by the integral of the same determinant over a subset of the solution variety scaled by an appropriate constant. Although the integral, in general, is hard to compute analytically, it can be easily estimated using Monte Carlo integration. To speed up the convergence of the Monte Carlo integration method, we specialize the general integral formula for square symmetric multi-beam cases (i.e., equal number of transmit and receive antennas and equal number of streams per user). Analogously, in the particular case of single-beam networks, we provide a combinatorial counting procedure that allows us to compute the exact number of solutions and analyze its asymptotic growth rate.

In addition to being of theoretical interest, the results proved in this work might also have some practical implications. For instance, finding scaling laws for the number of solutions with respect to the number of users could serve to analyze the asymptotic performance of linear IA, as discussed in \cite{Schmidt2010}, where information about the number of solutions is used to predict system performance when the best solution (or the best out of N) solutions is picked. Recent results \cite{Bresler2013} also suggest that the number of solutions is related to the computational complexity of designing the precoders and decoders satisfying the IA conditions.

The paper is organized as follows. In Section \ref{systemmodel}, the system model and the IA feasibility problem are briefly reviewed, paying special attention to the feasibility test in \cite{Gonzalez2012b} which is the starting point of this work. The main results of the paper are presented in Section \ref{sec:number}, where an integral formula, valid for arbitrary networks, for the number of IA solutions is given. Two special cases, square symmetric and single-beam networks, are analyzed in Section \ref{sec:special_cases}. A short review on Riemmanian manifolds and other mathematical results that will also be used during the derivations as well as the proofs of the main theorems in Section \ref{sec:number} are relegated to appendices. Numerical results are included in Section \ref{sec:numerical_experiments}.


\section{System model and background material}
\label{systemmodel}

In this section we describe the system model considered in the paper, introduce the notation, define the main algebraic sets used throughout the paper, and briefly review the feasibility conditions of linear IA problems for arbitrary wireless networks.

\subsection{Linear IA}
We consider the $K$-user MIMO interference channel with transmitter $k$ having $M_k\geq1$ antennas and receiver $k$ having $N_k\geq1$ antennas. Each user $k$ wishes to send $d_k\geq0$ streams or messages. We adhere to the notation used in \cite{Yetis10} and denote this (fully connected) asymmetric interference channel as $\prod_{k=1}^K \left(M_k\times N_k,d_k\right)=\left(M_1\times N_1,d_1\right)\cdots \left(M_K\times N_K,d_K\right)$. The symmetric case in which all users transmit $d$ streams and are equipped with $M$ transmit and $N$ receive antennas is denoted as $\left(M\times N,d\right)^K$. In the square symmetric case all users have the same number of antennas at both sides of the link $M=N$. In this paper we focus on the fully connected interference channel and, consequently, the number of interfering links will be $K(K-1)$.

User $j$ encodes its message using an $M_j \times d_j$ precoding matrix $V_j$ and the received signal is given by
\begin{equation}
\label{eq:received}
y_j=H_{jj}V_jx_j + \sum_{i\neq j}H_{ji}V_ix_i +n_j, \hspace{1cm} 1\leq j \leq K
\end{equation}
where $x_j$ is the $d_j \times 1$ transmitted signal and $n_j$ is the zero mean unit variance circularly symmetric additive white Gaussian noise vector. The MIMO channel from transmitter $l$ to receiver $k$ is denoted as $H_{kl}$ and assumed to be flat-fading and constant over time. Each $H_{kl}$ is an $N_k\times M_l$ complex matrix with independent entries drawn from a continuous distribution. The first term in (\ref{eq:received}) is the desired signal, while the second term represents the interference space. The receiver $j$ applies a linear decoder $U_j$ of dimensions $N_j \times d_j$, i.e.,
\begin{equation}
\label{eq:received1}
U_j^T y_j=U_j^T H_{jj}V_jx_j + \sum_{i\neq j}U_j^T H_{ji}V_ix_i +U_j^T n_j, \hspace{1cm} 1\leq j \leq K,
\end{equation}
where superscript $T$ denotes transpose.

The interference alignment (IA) problem is to find the decoders and precoders, $V_j$ and $U_j$, in such a way that the interfering signals at each receiver fall into a reduced-dimensional subspace and the receivers can then extract the projection of the desired signal that lies in the interference-free subspace. To this end it is required that the polynomial equations
\begin{equation}\label{eq:1}
U_k^TH_{kl}V_l=0,\qquad k\neq l,
\end{equation}
are satisfied, while the signal subspace for each user must be linearly independent of the interference subspace and must have dimension $d_k$, that is
\begin{equation}\label{eq:rank}
\rank(U_k^TH_{kk}V_k)=d_k,\qquad\forall\;k.
\end{equation}
We recall that all matrices $H_{kl}$ (including direct link matrices, $H_{kk}$) are generic, that is, their entries are independently drawn from a
continuous probability distribution. Consequently, once \eqref{eq:1} holds, \eqref{eq:rank} is satisfied almost surely if $U_k$ and $V_k$ are of maximal rank.
\subsection{Feasibility of IA: a brief review}

The IA feasibility problem amounts to study the relationship between $d_j,M_j,N_j,K$ such that the linear alignment problem is feasible. If the problem is feasible, the tuple $\left( d_1,\ldots,d_K \right)$ defines the degrees of freedom (DoF) of the system, that is the maximum number of independent data streams that can be transmitted without interference in the channel. The IA feasibility problem and the closely related problem of finding the maximum DoF of a given network have attracted a lot of research over the last years. For instance, the DoF for the 2-user and, under some conditions, for the symmetric $K$-user MIMO interference channel have been found in \cite{JafarDoF07} and \cite{JafarDoF10}, respectively. In this work we make the following assumptions:
\begin{equation}\label{eq:8}
1\leq d_k\leq N_k,\; \hspace{0.1cm}  \forall\;k,\qquad 1\leq d_l\leq M_l,\; \hspace{0.1cm}\forall\;l,
\end{equation}
and
\begin{equation}\label{eq:9}
d_kd_l<N_kM_l,\qquad\forall\;k \neq l,
\end{equation}
which are necessary conditions for feasibility which arise from the fact that two users of an interference channel cannot
reach their point-to-point bounds simultaneously since they have to leave at least a one-dimensional subspace for the interference.

The IA feasibility problem has also been intensively investigated in \cite{Yetis10,Razaviyayn1,BreslerCartwrightTseToappear,Slock10,LauISIT12}. In the following we make a short review of the main feasibility result presented in \cite{Gonzalez2012b}, which forms the starting point of this work.

We start by describing the three main algebraic sets involved in the feasibility problem which were first introduced in \cite{BreslerCartwrightTseToappear}:
\begin{itemize}
\item Input space formed by the MIMO matrices, which is formally defined as
\begin{equation} \label{eq:input}
\mathcal{H}=\prod_{k\neq l}\mathcal{M}_{N_k\times M_l}(\C)
\end{equation}
where $\prod$ holds for Cartesian product, and $\mathcal{M}_{N_k\times M_l}(\C)$ is the set of $N_k\times M_l$ complex matrices. Note that in \cite{OscarISIT12,Gonzalez2012b}, we let $\CH$ be the product of projective spaces instead of the product of affine spaces. The use of affine spaces is more convenient for the purposes of root counting.

\item Output space of precoders and decoders (i.e., the set where the possible outputs exist)

\begin{equation}\label{eq:output}
\mathcal{S}=\left(\prod_{k}\G{d_k}{N_k}\right)\times\left(\prod_{l}\G{d_l}{M_l}\right),
\end{equation}
where $\G{a}{b}$ is the Grassmannian formed by the linear subspaces of (complex) dimension $a$ in $\C^b$.

\item The solution variety, which is given by

\begin{equation}\label{eq:solution}
\CV=\{(H,U,V)\in\mathcal{H}\times\mathcal{S}:\text{ (\ref{eq:1}) holds}\}
\end{equation}

where $H$ is the collection of all matrices $H_{kl}$ and, similarly, $U$ and $V$ denote the set of $U_k$ and $V_l$, respectively. The set $\CV$ is given by certain polynomial equations, linear in each of the $H_{kl},U_k,V_l$ and therefore is an algebraic subvariety of the product space $\mathcal{H}\times\mathcal{S}$. Let us remind here that the IA equations given by (\ref{eq:1}) hold or do not hold independently of the particular chosen affine representatives of $U,V$.

\end{itemize}


The following diagram, illustrating the sets and the main projections involved in the feasibility problem, was considered in \cite{BreslerCartwrightTseToappear}:
\begin{equation}\label{eq:diag}
\begin{matrix}
&&\CV&&\\
\pi_1&\swarrow&&\searrow&\pi_2\\
\mathcal{H}&&&&\mathcal{S}
\end{matrix}
\end{equation}
Note that, given $H\in\mathcal{H}$, the set $\pi_1^{-1}(H)$ is a copy of the set of $U,V$ such that (\ref{eq:1}) holds, that is the solution set of the linear interference alignment problem. On the other hand, given $(U,V)\in\mathcal{S}$, the set $\pi_2^{-1}(U,V)$ is a copy of the set of $H\in\mathcal{H}$ such that (\ref{eq:1}) holds.

The feasibility question can then be restated as, {\em is $\pi_1^{-1}(H)\neq\emptyset$ for a generic $H$?} Following this formulation, the problem was first tackled in \cite{BreslerCartwrightTseToappear} and \cite{Razaviyayn1} where some necessary and sufficient conditions were given. Analytical expressions were limited to some symmetric scenarios of interest.
In \cite{Gonzalez2012b}, a solution to this problem was given by proposing a probabilistic polynomial time feasibility test for completely arbitrary interference channels. The test exploited the fact that system is feasible if and only if two conditions are fulfilled:
\begin{enumerate}
\item The algebraic dimension of $\CV$ must be larger than or equal to the dimension of $\mathcal{H}$, i.e.,
\begin{equation}\label{eq:2}
s=\left( \sum_{k} d_k(N_k+ M_k -2d_k) \right)-\left(\sum_{k \neq l}d_kd_l\right) \geq 0.
\end{equation}
In other words this condition means that, for the problem of polynomial equations to have a solution, the total number of variables must be larger than or equal to the total number of equations ($s\geq 0$).
We recall that a more general version of this condition was first established in \cite{Yetis10}. In that work, an interference channel was classified as \textit{proper} when the number of variables was larger than or equal to the number of equations for every subset of equations. Otherwise, it was classified as improper. More recently, in \cite{Razaviyayn1} it was rigorously proved that improper systems are always infeasible which implies that a system with $s<0$ is infeasible. 

\item For {\it some} element $(H,U,V) \in \CV$, the linear mapping
\begin{equation}\label{eq:4}
\begin{matrix}
\theta:&\left(\prod_{k}\mathcal{M}_{N_k\times d_k}(\C)\right)\times \left(\prod_{l}\mathcal{M}_{M_l\times d_l}(\C)\right)&\ra &\prod_{k \neq l}\mathcal{M}_{d_k\times d_l}(\C)\\
&(\{\dot{U}_k\},\{\dot{V}_l\})&\mapsto& \left\{ \dot{U}_k^T H_{kl}{V}_l+{U}_k^T H_{kl}\dot{V}_l\right\}_{k \neq l}
\end{matrix}
\end{equation}
is surjective, i.e., it has maximal rank equal to $\sum_{k\neq l}d_kd_l$. This condition amounts to saying that the projection from the tangent plane at an arbitrary point of the solution variety to the tangent plane of the input space must be surjective: that is, one tangent plane must cover the other. Moreover, in this case, the mapping (\ref{eq:4}) is surjective for {\em almost every} $(H,U,V)\in\CV$.
\end{enumerate}

We recall that these conditions were essentially found in \cite{Razaviyayn1} and \cite{BreslerCartwrightTseToappear} by using different mathematical tools than the ones used in \cite{Gonzalez2012b}. In this paper we will build on the results in \cite{Gonzalez2012b} using as a starting point the result stating that, when a system is feasible and $s=0$, the number of IA solutions is finite and constant for almost all channel realizations. This is formally stated in the following lemma.

\begin{lemma}[See Th. 1 in \cite{Gonzalez2012b}]\label{lem:constant}
For a feasible scenario and for almost every $H$, the solution set is a smooth complex algebraic submanifold of dimension $s$. If $s=0$, then there is constant $C\geq 1$ such that for every choice of $H_{kl}$ out of a proper algebraic subvariety (thus, for every choice out of a zero measure set) the system has exactly $C$ aligment solutions. 
\end{lemma}

\begin{IEEEproof}
See \cite[Section V]{Gonzalez2012b}.
\end{IEEEproof}

\section{The number of solutions of feasible IA problems}

\label{sec:number}


\subsection{Preliminaries}
 As shown in \cite{OscarISIT12,Gonzalez2012b}, the surjectivity of the mapping $\theta$ in (\ref{eq:4}) can easily checked by a polynomial-complexity test that can be applied to arbitrary $K$-user MIMO interference channels. The test basically consists of two main steps: i) to find an {\it arbitrary} point in the solution variety and ii) to check the rank of a matrix constructed from that point. 
As a solution to the first step we follow \cite[Sec. IV]{Gonzalez2012b} and choose a simple solution to the IA equations. Specifically, we take structured channel matrices given by
\begin{equation}\label{eq:form1}
 H_{kl}=\begin{pmatrix}0_{d_k\times d_l}&A_{kl}\\B_{kl}&C_{kl}\end{pmatrix},
\end{equation}
with precoders and decoders given by
\begin{equation}\label{eq:form2}
V_l=\binom{I_{d_l}}{0_{(M_l-d_l)\times d_l}}, \hspace{0.5cm}  U_k=\binom{I_{d_k}}{0_{(N_k-d_k)\times d_k}},
\end{equation}
which trivially satisfy $U_k^T H_{kl}V_l=0$ and therefore belong to the solution variety. We claim that essentially all the useful information about $\CV$ can be obtained from the subset of $\CV$ consisting of the triples $(H_{kl},U_k,V_l)$ where its elements have the form (\ref{eq:form1}) and (\ref{eq:form2}). In order to see this, we pick any other element $(\tilde{H}_{kl},\tilde{U}_k,\tilde{V}_l)\in\CV$.  Without loss of generality we can assume $\tilde{U}_k$ and $\tilde{V}_l$ lie in the Stiefel manifold i.e. they satisfy $\tilde{U}_k^*\tilde{U}_k=I$ and $\tilde{V}_l^*\tilde{V}_l=I$ where the superscript $*$ denotes Hermitian (conjugate transpose). Now, we will show how this element of $\CV$ can be converted into one of the form (\ref{eq:form1}) and (\ref{eq:form2}). First, we compute a QR decomposition of $\tilde{U}_k$ and $\tilde{V}_l$, that is
\[
\tilde{U}_k=P_k\binom{I_{d_k}}{0_{(N_k-d_k)\times d_k}}=P_kU_k,\quad \tilde{V}_l=Q_l\binom{I_{d_l}}{0_{(M_l-d_l)\times d_l}}=Q_lV_l.
\]
where $P_k$ and $Q_l$ are unitary matrices. Then, the IA condition can be written as
\[
\tilde{U}_k^T\tilde{H}_{kl}\tilde{V}_l=U_k^TP_k^T\tilde{H}_{kl}Q_lV_l=0.
\]
It is now clear that the transformed channels $H_{kl}=P_k^T\tilde{H}_{kl}Q_l$ have the form (\ref{eq:form1}), and the transformed precoders $V_l$ and decoders $U_k$ have the form (\ref{eq:form2}). We have just described an isometry that sends $(\tilde{H}_{kl},\tilde{U}_k,\tilde{V}_l)$ to $(H_{kl},U_k,V_l)$. The situation is thus similar to that of a torus: every point can be sent to some predefined vertical circle through a rotation, thus the torus is essentially understood by ``moving'' a circumference and keeping track of the visited places. The same way, $\CV$ can be thought of as moving the set of triples of the form (\ref{eq:form1}) and (\ref{eq:form2}), and keeping track of the visited places. Technically, $\CV$ is the orbit of the set of triples of the form (\ref{eq:form1}) and (\ref{eq:form2}) under the isometric action of a product of unitary groups.

In \cite{Gonzalez2012b} this idea is rigorously exploited, proving that, for the purpose of checking feasibility or counting solutions, we can replace the set of arbitrary complex matrices $\CH$ by the set of structured matrices
\begin{equation}\label{eq:H_I}
\CHI=\prod_{k\neq l}\begin{pmatrix}0_{d_k\times d_l}&A_{kl}\\B_{kl}&C_{kl}\end{pmatrix}\equiv\pi_2^{-1}\left(\left\lbrace\binom{I_{d_k}}{0_{(N_k-d_k)\times d_k}}\right\rbrace_k,\left\lbrace\binom{I_{d_l}}{0_{(M_l-d_l)\times d_l}}\right\rbrace_l\right).
\end{equation}
The mapping $\theta$ in (\ref{eq:4}) has a simpler form for triples of the form (\ref{eq:form1}) and (\ref{eq:form2}), and can be replaced by a new mapping $\Psi$ defined as
\begin{equation}\label{eq:4bis}
\begin{matrix}
\Psi:&\left(\prod_{k}\mathcal{M}_{(N_k-d_k)\times d_k}(\C)\right)\times \left(\prod_{l}\mathcal{M}_{(M_l-d_l)\times d_l}(\C)\right)&\ra &\prod_{k \neq l}\mathcal{M}_{d_k\times d_l}(\C)\\
&(\{\dot{{U}}_k\}_{k},\{\dot{{V}}_l\}_{l})&\mapsto&\left(\dot{{U}}_k^T {B}_{kl}+{A}_{kl}\dot{{V}}_l\right)_{k\neq l}
\end{matrix}.
\end{equation}
We remark that, since the mapping \eqref{eq:4bis} is linear in both
$\dot{U}_k$ and $\dot{V}_l$, it can be represented by a matrix. With a slight abuse of notation we will use the symbol $\Psi$ to refer to both the mapping and the matrix representing that mapping. In this paper, we will be interested in the function $\det(\Psi\Psi^*)$, which depends on the channel realization $H$ through the blocks ${A}_{kl}$ and ${B}_{kl}$ only. The dimensions of $\Psi$ are $\sum_{k\ne l}d_kd_l\times \sum_{k=1}^K(M_k+N_k-2d_k)d_k$. In the particular case of $s=0$, the one of interest for this paper, $\Psi$ is a square matrix of size $\sum_{k\ne l}d_kd_l$ and, therefore, $\det(\Psi\Psi^*)=|\det(\Psi)|^2$. The interested reader can find additional details on the structure of the matrix $\Psi$ in \cite{Gonzalez2012b} and in the example in Section \ref{sec:example} below.


\subsection{Main results}

We use the following notation: given a Riemannian manifold $X$ with total finite volume denoted as $Vol(X)$ (the volume of the manifolds used in this paper are reviewed in Appendix \ref{appendix:preliminaries}), let
\[
\dashint_{x\in X}f(x)\,dx=\frac{1}{Vol(X)}\int_{x\in X}f(x)\,dx
\]
be the average value of a integrable function $f:X\ra\R$. Fix $d_j,M_j,N_j$ and $\Phi$ satisfying (\ref{eq:8}) and (\ref{eq:9}) and let $s$ be defined as in (\ref{eq:2}). The main results of the paper are Theorems 1 and 2 below, which give integral expressions for the number of IA solutions when $s=0$ which is denoted as $\#(\pi_1^{-1}(H_0))$. For the sake of rigorousness, we denote a generic channel realization as $H_0$. Recall that the particular choice of $H_0$ is irrelevant since the number of solutions is the same for all channel realizations out of some zero-measure set.

\begin{theorem}\label{th:main1}
Assume that $s=0$, and let $\CH_\epsilon\subseteq\CH$ be any open set such that the following holds: if $H=(H_{kl})\in\CH_\epsilon$ and $P_k,Q_k$, $1\leq k\leq K$ are unitary matrices of respective sizes $N_k,M_k$, then
\[
(P_k^TH_{kl}Q_l)\in\CH_\epsilon.
\]
(We may just say that $\CH_\epsilon$ is invariant under unitary transformations). Then, for every $H_0\in\CH$ out of some zero--measure set, we have:
\begin{align}\label{eq:main1}
\#(\pi_1^{-1}(H_0))=C\int_{H\in \CHI\cap\CH_\epsilon}\det(\Psi\Psi^*)\,dH,
\end{align}
where
\[
C=\frac{Vol(\CS)}{Vol(\CH_\epsilon)},
\]
with $\CS$ being the output space (Cartesian product of Grasmannians) in Eq. (\ref{eq:output}) and $\CHI$ defined in \eqref{eq:H_I}.
\end{theorem}

\begin{IEEEproof}
See Appendix \ref{sec:proof1}.
\end{IEEEproof}

If we take $\CH_\epsilon$ to be the set
\[
\{(H_{kl}):\|H_{kl}\|_F\in(1-\epsilon,1+\epsilon)\}
\]
(with $\|\cdot\|_F$ denoting Frobenius norm) and we let $\epsilon\rightarrow0$ we get:
\begin{theorem}\label{th:main2}
For an interference channel with $s=0$, and for every $H_0\in\CH$ out of some zero--measure set, we have:
\[
\#(\pi_1^{-1}(H_0))=C\dashint_{H\in \CHI,\|H_{kl}\|_F=1 }\det(\Psi\Psi^*)\,dH,
\]
where
\[
C=\prod_{k \neq l}\left(\frac{\Gamma(N_kM_l)}{\Gamma(N_kM_l-d_kd_l)}\right)\times
\]
\[
\prod_{k}\left(\frac{\Gamma(2)\cdots \Gamma(d_k)\cdot  \Gamma(2)\cdots \Gamma(N_k-d_k)}{\Gamma(2)\cdots \Gamma(N_k)}\right)\times
\]
\[
 \prod_{l}\left(\frac{\Gamma(2)\cdots \Gamma(d_l)\cdot  \Gamma(2)\cdots \Gamma(M_l-d_l)}{\Gamma(2)\cdots \Gamma(M_l)}\right)
\]
\end{theorem}

\begin{IEEEproof}
See Appendix \ref{sec:proof2}.
\end{IEEEproof}

\begin{remark}
As proved in \cite{Gonzalez2012b}, if the system is infeasible then $\det(\Psi\Psi^*)=0$ for every choice of $H,U,V$ and hence Theorem \ref{th:main1} still holds. On the other hand, if the system is feasible and $s>0$ then there is a continuous of solutions for almost every $H_{kl}$ and hence it is meaningless to count them (the value of the integrals in our theorems is not related to the number of solutions in that case). Note also that the equality of Theorem \ref{th:main1} holds for every unitarily invariant open set $\CH_\epsilon$, which from Lemma \ref{lem:constant} implies that the right--hand side of (\ref{eq:main1}) has the same value for all such $\CH_\epsilon$.
\end{remark}

\subsection{Example: the $(2 \times 2,1 )^3$ system}\label{sec:example}
\label{ex:example_th2}
In this example we specialize Theorem \ref{th:main2} to the $(2 \times 2,1 )^3$ scenario. Although the number of IA solutions for this network is known to be $2$ from the seminal work \cite{Jafar08}, this example will serve to illustrate the main steps followed to find the solution of the integral equation, and the difficulties to extend this analysis to more complex scenarios.

Let us start by considering structured $(2 \times 2)$ matrices of the form
\begin{equation}\label{eq:matricesapp}
 \bar{H}_{kl}=\begin{pmatrix}0&A_{kl}\\B_{kl}&C_{kl}\end{pmatrix},
\end{equation}
whose entries, without loss of generality, can be taken as independent complex normal random variables with zero mean and variance 2: $A_{kl}\sim CN(0,2)$, $B_{kl} \sim CN(0,2)$ and $C_{kl}\sim CN(0,2)$\footnote{The real and imaginary parts of each entry are independent real Gaussian random variables with zero mean and variance 1}. Each one of these random matrices is now normalized to get
\begin{equation}\label{eq:matricesapp_norm}
 H_{kl}= \begin{pmatrix}0&A_{kl}/\|\bar{H}_{kl}\|_F\\B_{kl}/\|\bar{H}_{kl}\|_F&C_{kl}/\|\bar{H}_{kl}\|_F\end{pmatrix}.
\end{equation}
The collection of matrices generated in this way is uniformly distributed on the set $\{ \CHI \bigcap \|H_{kl}\|_F=1 \}$ in Theorem \ref{th:main2}. Therefore, the integral formula given in Theorem 2 yields:
\begin{equation}\label{eq:expectation}
\sharp(\pi_1^{-1}(H_0)) = C \, E\left[ \det(\Psi \Psi^*) \right]=C \, E\left[|\det(\Psi)|^2\right],
\end{equation}
where $C=3^6=729$.

Choosing a natural order in the image space, the $6\times 6$ matrix $\Psi$ defining the mapping for the $(2 \times 2,1)^3$ scenario is 
\begin{equation}
\Psi = \left[ \begin{array}{cccccc}
B_{12}/\|\bar{H}_{12}\|_F&0&0&0&A_{12}/\|\bar{H}_{12}\|_F&0\\
B_{13}/\|\bar{H}_{13}\|_F&0&0&0&0&A_{13}/\|\bar{H}_{13}\|_F\\
0&B_{21}/\|\bar{H}_{21}\|_F&0&A_{21}/\|\bar{H}_{21}\|_F&0&0\\
0&B_{23}/\|\bar{H}_{23}\|_F&0&0&0&A_{23}/\|\bar{H}_{23}\|_F\\
0&0&B_{31}/\|\bar{H}_{31}\|_F&A_{31}/\|\bar{H}_{31}\|_F&0&0\\
0&0&B_{32}/\|\bar{H}_{32}\|_F&0&A_{32}/\|\bar{H}_{32}\|_F&0
\end{array}\right].
\nonumber
\end{equation}
It is easy to compute the determinant of this matrix expanding it along the first column:
\[
\det(\Psi)=\frac{B_{12}A_{13}A_{32}B_{23}B_{31}A_{21}}{\|\bar{H}_{12}\|_F\|\bar{H}_{13}\|_F\|\bar{H}_{32}\|_F\|\bar{H}_{23}\|_F\|\bar{H}_{31}\|_F\|\bar{H}_{21}\|_F}-\frac{B_{13}A_{12}A_{23}B_{21}A_{31}B_{32}}{\|\bar{H}_{13}\|_F\|\bar{H}_{12}\|_F\|\bar{H}_{23}\|_F\|\bar{H}_{21}\|_F\|\bar{H}_{31}\|_F\|\bar{H}_{32}\|_F}.
\]
Therefore,
\begin{align*}
|\det(\Psi)|^2&=\left|\frac{B_{12}A_{13}A_{32}B_{23}B_{31}A_{21}}{\|\bar{H}_{12}\|_F\|\bar{H}_{13}\|_F\|\bar{H}_{32}\|_F\|\bar{H}_{23}\|_F\|\bar{H}_{31}\|_F\|\bar{H}_{21}\|_F}\right|^2\\
&+\left|\frac{B_{13}A_{12}A_{23}B_{21}A_{31}B_{32}}{\|\bar{H}_{13}\|_F\|\bar{H}_{12}\|_F\|\bar{H}_{23}\|_F\|\bar{H}_{21}\|_F\|\bar{H}_{31}\|_F\|\bar{H}_{32}\|_F}\right|^2\\
&-2\mathcal{R}e\left(\frac{B_{12}A_{13}A_{32}B_{23}B_{31}A_{21}B_{13}A_{12}A_{23}B_{21}A_{31}B_{32}}{(\|\bar{H}_{12}\|_F\|\bar{H}_{13}\|_F\|\bar{H}_{32}\|_F\|\bar{H}_{23}\|_F\|\bar{H}_{31}\|_F\|\bar{H}_{21}\|_F)^2}\right).
\end{align*}
The first of these quantities is the product of $6$ i.i.d. random variables, thus
\[
E\left[\left|\frac{B_{12}A_{13}A_{32}B_{23}B_{31}A_{21}}{\|\bar{H}_{12}\|_F\|\bar{H}_{13}\|_F\|\bar{H}_{32}\|_F\|\bar{H}_{23}\|_F\|\bar{H}_{31}\|_F\|\bar{H}_{21}\|_F}\right|^2\right]=E\left[\left|\frac{B_{12}}{\|\bar{H}_{12}\|_F}\right|^2\right]^6.
\]
Similarly,
\[
E\left[\left|\frac{B_{13}A_{12}A_{23}B_{21}A_{31}B_{32}}{\|\bar{H}_{13}\|_F\|\bar{H}_{12}\|_F\|\bar{H}_{23}\|_F\|\bar{H}_{21}\|_F\|\bar{H}_{31}\|_F\|\bar{H}_{32}\|_F}\right|^2\right]=E\left[\left|\frac{B_{12}}{\|\bar{H}_{12}\|_F}\right|^2\right]^6.
\]
Finally,
\[
E\left[\frac{B_{12}A_{13}A_{32}B_{23}B_{31}A_{21}B_{13}A_{12}A_{23}B_{21}A_{31}B_{32}}{(\|\bar{H}_{12}\|_F\|\bar{H}_{13}\|_F\|\bar{H}_{32}\|_F\|\bar{H}_{23}\|_F\|\bar{H}_{31}\|_F\|\bar{H}_{21}\|_F)^2}\right]=0,
\]
because $B_{12}$ has the same distribution as $-B_{12}$. That is, the isometry $B_{12}\mapsto -B_{12}$ changes the sign of the function inside the expectation symbol but the expectation is unchanged when multiplied by $-1$. Hence, the expectation is $0$. We have thus proved that
\[
\sharp(\pi_1^{-1}(H_0)) = 2\cdot 3^6 E\left[\left|\frac{B_{12}}{\|\bar{H}_{12}\|_F}\right|^2\right]^6.
\]
We now compute the last term using the fact that $\left|\frac{B_{12}}{\|\bar{H}_{12}\|_F}\right|^2 \sim  {\rm Beta}(1,2)$, where ${\rm Beta}(1,2)$ denotes a beta-distributed random variable with shape parameters 1 and 2.


Consequently,
\[
\sharp(\pi_1^{-1}(H_0)) = 2\cdot 3^6\left(\frac{1}{3}\right)^6=2,
\]
as desired.

\subsection{Estimating the number of solutions via Monte Carlo integration}\label{sec:montecarlo}
Given the complexity of analytically computing the integral in Theorem \ref{th:main2} for general scenarios (as illustrated with a simple example in Section \ref{ex:example_th2}), we will provide, in this section, a method to approximate its value using Monte Carlo integration. Our main reference here is \cite{MonteCarlo}. The {\em Crude Monte Carlo} method for computing the average
\[
E_X(f)=\dashint_{x\in X}f(x)\,dx
\]
of a function $f$ defined on a finite-volume manifold $X$ consists just in choosing many points at random, say $x_1,\ldots,x_n$ for $n>>1$, uniformly distributed in $X$, and approximating
\begin{equation}\label{eq:mean}
\dashint_{x\in X}f(x)\,dx\approx E_n=\frac{1}{n}\sum_{j=1}^nf(x_j).
\end{equation}

The most reasonable way to implement this in a computer program is to write down an iteration that computes $E_1,E_2,E_3,\ldots$ The key question to be decided is how many such $x_j$ we must choose to get a reasonably good approximation of the integral. To do so, we follow the ideas in \cite[Sec. 5]{MonteCarlo}: first note that the random variable $Y_n=\sqrt{n}(E_X(f)-E_n)$ approaches, by the Central Limit Theorem, a Normal distribution, that is the density function of $Y_n$ can be approximated by
\[
\frac{1}{\sigma\sqrt{2\pi}}e^{-\frac{t^2}{2\sigma^2}},
\]
for some $\sigma$ which is actually the standard deviation of $f$, given by
\[
\sigma^2=\dashint_{x\in X}\left(f(x)-E_X(f)\right)^2\,dx.
\]
Now note that
\[
\frac{1}{\sigma\sqrt{2\pi}}\int_{-2\sigma}^{2\sigma}e^{-\frac{t^2}{2\sigma^2}}\,dt\underset{t=s\sigma}{=}\frac{1}{\sqrt{2\pi}}\int_{-2}^{2}e^{-\frac{s^2}{2}}\,dt=0.9544\ldots
\]
Namely, for any random variable $Y$ following a normal distribution $N(0,\sigma)$, we have $|Y|\leq2\sigma$ with probability greater than $0.95$. Note that the reasoning above is not a formal proof but a heuristic argument. First, $Y_n$ is not exactly normal but, for a large $n$, our approximation will still serve its purpose. Second, there exists no way to guarantee that the integral of a generic function is correctly computed by Monte Carlo methods, see \cite[Sec. 5]{MonteCarlo}.

In order to get an estimate, we need to approximate $\sigma$. The unbiased estimator of $\sigma$ is
\begin{equation}\label{eq:variance}
\sigma_n=\left(\frac{1}{n-1}\sum_{j=1}^n(f(x_j)-E_n)^2\right)^{1/2}.
\end{equation}
We thus have that, with probability greater than $0.95$,
\[
|E_X(f)-E_n|\lesssim\frac{2\sigma_n}{\sqrt{n}}.
\]
If we stop the iteration when $\frac{\sigma_n}{\sqrt{n}E_n}\leq\varepsilon$, then, with a probability of $0.95$ on the set of random sequences of $n$ terms, the relative error satisfies
\[
\frac{\left|E_X(f)-E_n\right|}{|E_n|}\lesssim 2\varepsilon.
\]
For example, if we stop the iteration when $\frac{\sigma_n}{\sqrt{n}E_n}\leq0.05$, then, we can expect to be making an error of about $10$ percent in our calculation of $E_X(f)$.
The whole procedure for a general system is illustrated in Algorithm \ref{alg:alg_numsol_general}, which is based on Theorem \ref{th:main2}. 

\begin{algorithm}[!t]
\DontPrintSemicolon
\SetAlgoVlined
\KwIn{Relative error, $\varepsilon$; number of antennas, $\{M_l\}$ and $\{N_k\}$; streams, $d_k$; and users, $K$.}
\KwOut{Approximate number of IA solutions, $E_n$.}
$n=1$\;
\Repeat(){$\frac{\sigma_n}{\sqrt{n}E_n}<\varepsilon$}{
Generate a set of random matrices $\{A_{kl}\}$, $\{B_{kl}\}$ and $\{C_{kl}\}$ with i.i.d. $CN(0,2)$ entries.\;
Build channel matrices $\{H_{kl}\}$ according to \eqref{eq:form1}.\;
Normalize every channel matrix $H_{kl}$ such that $\|H_{kl}\|_F=1$.\;
Build the matrix $\Psi$ defining \eqref{eq:4bis}.\;
Compute $D_n=C\det(\Psi \Psi^*)$ where $C$ is taken from Theorem \ref{th:main2}.\;
Calculate $E_n$ and $\sigma_n$ according to \eqref{eq:mean} and \eqref{eq:variance}, respectively, where $f(x_j)$ is now $D_j$.\;
$n=n+1$. 
}()
\caption{Computing the number of IA solutions for general scenarios $\prod_{k=1}^K \left(M_k\times N_k,d_k\right)$.}
\label{alg:alg_numsol_general}
\end{algorithm}

\section{Algorithmic aspects and special cases}
\label{sec:special_cases}
We have shown how Theorem \ref{th:main2} can be used to approximate the number of IA solutions of a given interference channel using Monte Carlo integration. Nevertheless, our numerical experiments demonstrate that the convergence of the integral is, in general, slow. In this section, with the aim of mitigating this problem, we provide specializations for two cases of interest: square symmetric and single-beam scenarios.

\subsection{The square symmetric case}
The so-called square symmetric case is that in which all the $d_k$ and all the $N_k$ and $M_k$ are equal for all $k$. Furthermore, we are restricted to $s=0$ (for the solution counting to be meaningful) and to $K\geq3$ (for IA to make sense); which implies $N=M\geq 2d$.
Under these assumptions, we can write another integral such that Monte Carlo integration has been experimentally observed to converge faster:

\begin{theorem}\label{th:main3}
Let us consider a symmetric {\it square} interference channel ($N_k=M_k=N$ and $d_k=d$, $\forall k$) with $s=0$. Assuming additionally that $K\geq3$, then for every $H_0\in\CH$ out of some zero--measure set, we have:
\[
\#\pi_1^{-1}(H_0)=\left(\frac{2^{d^2}Vol(\CU_{N-d})^2}{Vol(\CU_{N})Vol(\CU_{N-2d})}\right)^{K(K-1)}Vol(\CS)\dashint_{(A_{kl}^*,B_{kl}) \in \CU_{(N-d)\times d}^2}\det(\Psi \Psi^*)\;dH,
\]
where $\Psi$ is again defined by (\ref{eq:4bis}) and the input space of MIMO channels where we have to integrate are now
\[
H_{kl}=\begin{pmatrix}0_{d\times d}&A_{kl}\\B_{kl}&0_{(N-d)\times (N-d)}\end{pmatrix},
\]
whose blocks, $A_{kl}^*$ and $B_{kl}$, are matrices in the complex Stiefel manifold, denoted as $\CU_{(N-d)\times d}$, and formed by all the (ordered) collections of $d$ orthonormal vectors in  $\C^{(N-d)}$. On the other hand, $\CU_{a}$ denotes the unitary group of dimension $a$, whose volume can be found in Appendix \ref{appendix:preliminaries}.
\end{theorem}
\begin{IEEEproof}
See Appendix \ref{sec:proof3}.
\end{IEEEproof}

\begin{remark}
The value of the constant preceding the integral in Theorem \ref{th:main3} is (using that $2N-dK-d=0$ when $s=0$):
\[
C=\left(\frac{2^{d^2}Vol(\CU_{N-d})^2}{Vol(\CU_{N})Vol(\CU_{N-2d})}\right)^{K(K-1)}Vol(\CS)=
\]
\[
\left(\frac{\Gamma(N-d+1)\cdots\Gamma(N)}{\Gamma(N-2d+1)\cdots\Gamma(N-d)}\right)^{K(K-1)}\left(\frac{\Gamma(2)\cdots\Gamma(d)}{\Gamma(N-d+1)\cdots\Gamma(N)}\right)^{2K}
\]
\end{remark}

\begin{example}\label{ex:example_th3}
In this example we will use Theorem \ref{th:main3} to calculate the number of solutions for the scenario $(2\times 2,1)^3$ again. First, we calculate the value of the constant $C$ which happens to be equal to 1 and, consequently, the number of solutions is directly given by the average of the determinant. \footnote{Indeed, $C=1$ for all systems whenever $N=2d$ or, equivalently, $K=3$.}

Subsequent calculations are similar to those in the example in Section \ref{ex:example_th2}. The main difference is that, in this case, $A_{kl}$ and $B_{kl}$ are restricted to be elements of the complex Stiefel manifold, in this case, the unit-circle. Then, 
\[
\#(\pi^{-1}(H_0))=2E[|A_{12}|^2]^6=2.
\]
\end{example}
From Example \ref{ex:example_th3} it is clear that Theorem \ref{th:main3} has remarkably simplified the calculation of the integral by reducing the dimensionality of the integration domain. However, for larger scenarios we may still need to resort to the Monte Carlo integration procedure in Section \ref{sec:montecarlo} to approximate the integral in Theorem \ref{th:main3}. Algorithm \ref{alg:alg_numsol_square} summarizes the proposed method.

\begin{algorithm}[!t]
\DontPrintSemicolon
\SetAlgoVlined
\KwIn{Relative error, $\varepsilon$; number of antennas, $N$; streams, $d$; and users, $K$.}
\KwOut{Approximate number of IA solutions, $E_n$.}
$n=1$\;
\Repeat(){$\frac{\sigma_n}{\sqrt{n}E_n}<\varepsilon$}{
Generate a set of $(N-d)\times d$ matrices $\{A_{kl}^*\}$ and $\{B_{kl}\}$, independently and uniformly distributed in the Stiefel manifold.\;
Build the matrix $\Psi$ defining \eqref{eq:4bis}.\;
Compute $D_n=C\det(\Psi \Psi^*)$ where $C$ is taken from Theorem \ref{th:main3}.\;
Calculate $E_n$ and $\sigma_n$ according to \eqref{eq:mean} and \eqref{eq:variance}, respectively, where $f(x_j)$ is now $D_j$.\;
$n=n+1$. 
}()
\caption{Computing the number of IA solutions for symmetric square scenarios $(N\times N,d)^K$.}
\label{alg:alg_numsol_square}
\end{algorithm}

\subsection{The single-beam case}
\label{sec:single_beam}
The results of Theorems \ref{th:main1}, \ref{th:main2} and \ref{th:main3} are general and can be applied to systems where each user wishes to transmit an arbitrary number of streams. This subsection is devoted to specialize Theorem \ref{th:main2} to the particular case of single-beam MIMO networks (i.e. $d_k=1$, $\forall\,k$). First, we should mention that, from a theoretical point of view, the single-beam case was solved in \cite{Yetis10}, where it was shown that the number of IA solutions for single-beam feasible systems matches the mixed volume of the Newton polytopes that support each equation of the system\footnote{This is not true for multibeam cases because, in this case, the genericity of the system of equations is lost.}. However, from a practical point of view, the computation of the mixed volume of a set of bilinear equations using the available software tools \cite{Lee07mixedvolume} can be very demanding. As a consequence, the exact number of IA solutions is only known for some particular cases \cite{Yetis10,Schmidt2010}.

\begin{theorem}\label{th:numsol_singlebeam}
The number of IA solutions for an arbitrary single beam scenario with $s=0$ is given by 
\begin{equation}\label{eq:numsols_singlebeam_closedform}
\#(\pi_1^{-1}(H_0)) = \frac{\per(T)}{\prod_k(N_k-1)!\prod_l(M_l-1)!}
\end{equation}
where $T$ is the matrix built by replacing the non-zero elements of $\Psi$ by ones and $\per(T)$ denotes its permanent. 

Equivalently, 
\begin{equation}\label{eq:numsols_singlebeam_closedform_as_numtables}
\#(\pi_1^{-1}(H_0)) = \#\mathcal{A}^*(R,C)
\end{equation}
where $R=(K-N_1,\ldots,K-N_K)$, $C=(M_1-1,\ldots,M_K-1)$ and $\#\mathcal{A}^*(R,C)$ denotes the number of elements in $\mathcal{A}^*(R,C)$ which is defined as the class of zero-trace $K\times K$ binary matrices with row sums $R$ and column sums $C$.  
\end{theorem}
\begin{IEEEproof}
See Appendix \ref{sec:proof4}.
\end{IEEEproof}
In spite of its apparent simplicity, evaluating \eqref{eq:numsols_singlebeam_closedform} may be very hard. In fact, computing the permanent is, in general, proven to be \#P-complete \cite{Valiant1979} even for (0,1)-matrices where \#P is defined as the class of functions that count the
number of solutions in an NP problem. 

On the other hand, \eqref{eq:numsols_singlebeam_closedform_as_numtables} establishes an equivalence between the problem of computing the number of solutions of single-beam scenarios and the problem of counting the number of zero-trace binary matrices with prescribed rows and column sums. From a practical point of view, the result in \eqref{eq:numsols_singlebeam_closedform_as_numtables} suggests that the IA problem can be interpreted as transmitters and receivers collaborating to cancel every single interfering link. A transmitter zero-forcing a link is encoded as a one in $S$ whereas a receiver zero-forcing a link is encoded   
as a zero. The total number of possible collaboration strategies gives the number of IA solutions.

Unfortunately, calculating $\#\mathcal{A}^*(R,C)$ is a non-trivial particular case of a problem which is also known to be \#P-complete \cite[Theorem 9.1]{Diaconis1995}. For the interested reader, we have computed several exact values which are compiled in Table \ref{tab:num_solutions_1beam_exact_approx}. Our algorithm performs a recursive tree search, commonly known as \textit{backtracking} \cite{Golomb_Baumert_1965} and is summarized in Algorithm \ref{alg:backtracking}.

\begin{algorithm}[!t]
\DontPrintSemicolon
\SetAlgoVlined
\SetKwProg{Fn}{function}{}{} 
\SetKwFunction{Candidate}{get\_candidates}
\SetKwFunction{Backtrack}{backtrack}
\KwIn{Number of antennas, $\{M_k\}$ and $\{N_k\}$; and users, $K$. }
\KwOut{Number of solutions, $S$.}
$S=0$ \tcp*{No solutions found yet}
$table={\bf 0}$ \tcp*{Empty K$\times$K table to fill with 1s}
$row=0$, $col=0$ \tcp*{Row and column indexes}
$S=$ \Backtrack{$table$, $row$, $col$, $S$}\;
\hrule
\BlankLine
\Fn{$S=$ \Backtrack{$table$, $row$, $col$, $S$}}{
\eIf{table is a valid solution}{
$S=S+1$ \tcp*{Valid solution found}
}{
 \ForEach{$(row,col)$ in \Candidate{$table$,$row$,$col$}}{
$table(row,col)=1$ \tcp*{Fill the cell with a 1}
\Backtrack{$table$, $row$, $col$, $S$} \tcp*{Recursive call}
$table(row,col)=0$ \tcp*{Remove the 1}
}
}
\KwRet{S}\;
}
\hrule
\BlankLine
\Fn{$((crow_1,ccol_1),\ldots,(crow_N,ccol_N))=$ \Candidate{$table$,$row$,$col$}}{
\KwRet{list of candidate cells to store the next 1}\;
}
\caption{Backtracking procedure for counting the number of IA solutions for arbitrary single-beam scenarios $\prod_{k=1}^K \left(M_k\times N_k,1\right)$.}
\label{alg:backtracking}
\end{algorithm}

\subsubsection{Connections with graph theory problems}
\label{sec:connection_graph_theory}
For the particular case of symmetric $(M\times N, 1)^K$ scenarios, the IA solution counting problem can be restated as several well-studied combinatorial and graph theory problems. Most of these problems have been of historical interest and hence a lot of research has been done on them. Specifically, when the matrices in $\mathcal{A}^*(R,C)$ are seen as the adjacency matrix of a graph some connections to graph theory problems arise. It is natural, then, to find out that the number of solutions for some scenarios have already been computed in the literature:

\begin{itemize}
\item The number of solutions for $(2\times (K-1),1)^K$ scenarios is given by the number of derangements (permutations of $K$ elements with no fixed points), also known as \textit{rencontres numbers} or subfactorial. It is also the number of simple loop-free labeled $1$-regular digraphs with $K$ nodes. Interestingly, as demonstrated in \cite{OEIS}\cite[p.195]{Graham1994}, a closed-form solution is available:
$$\operatorname{round}\left(\cfrac{K!}{e}\right)\, .$$

\item The number of solutions for $(3\times (K-2),1)^K$ systems matches the number of simple loop-free labeled $2$-regular digraphs with $K$ nodes. In this case, a closed-form expression is also available \cite{OEIS}:

$$\sum_{k=0}^{K} \sum_{s=0}^{k} \sum_{j=0}^{K-k}\frac{(-1)^{k+j-s}K!(K-k)!(2K-k-2j-s)!}{s!(k-s)!((K-k-j)!)^{2}j!2^{2K-2k-j}}\, .$$

\item  In general, the number of solutions for the $(M\times (K-M+1),1)^K$ scenario matches the number of simple loop-free labeled $(M-1)$-regular digraphs with $K$ nodes. However, as far as we are aware, additional closed-form expressions do not exist.
\end{itemize}

\subsubsection{Bounds and asymptotic rate of growth}
\label{sec:bounds_growth_rate}
In order to derive appropriate bounds for the number of solutions it is convenient to go back to \eqref{eq:numsols_singlebeam_closedform} and apply some classical combinatorial results to bound the value of $\operatorname{per}(T)$. Herein, we will focus on symmetric systems: $(M\times N, 1)^K$. B\'{e}rgman's Theorem \cite[Theorem 7.4.5]{Brualdi1991} gives an upper bound for the permanent of an arbitrary matrix as a function of its row sums, $r_i$. In our case, every row (and column) sum is $K-1$ and the bound simplifies quite notably:
\begin{equation}\label{eq:bregman_theorem}
\operatorname{per}(T)\leq \prod_i^{K(K-1)}(r_i!)^{1/r_i}=((K-1)!)^K.
\end{equation} 
Additionally, we can use the fact that $T'=T/(K-1)$ is doubly stochastic to apply van der Waerden's conjecture (now proven) \cite{Egorycev1980,Falikman1981}, i.e. $\operatorname{per}(T')\geq n!/n^n$, where $n$ denotes the size of the matrix:
\begin{equation}\label{eq:vanderwaerden_theorem}
\operatorname{per}(T) = (K-1)^{K(K-1)}\operatorname{per}(T')\geq \frac{(K(K-1))!}{K^{K(K-1)}}.
\end{equation}
From the previous bounds and \eqref{eq:numsols_singlebeam_closedform}, the number of solutions is shown to be bounded above and below as follows:
\begin{equation}\label{eq:numsols_singlebeam_sandwich}
L=\frac{(K(K-1))!}{((M-1)!)^K ((N-1)!)^K K^{K(K-1)}}\leq \#(\pi_1^{-1}(H_0))\leq {{K-1} \choose {M-1}}^K=U.
\end{equation} 
Now, we study the growth rate of the number of solutions when the number of users increases. As a first step, we approximate every factorial in both bounds applying Stirling's formula, i.e. $\log(n!)\approx n \log n$ for large $n$. Interestingly, this approximation demonstrates that both upper and lower bounds are asymptotically equivalent
\begin{equation}\label{eq:rate_of_growth}
\log L\approx \log U \approx K(K-1) \log \frac{K-1}{K-M} + K(M-1) \log \frac{K-M}{M-1},  
\end{equation}
and the actual number of solutions, which is bounded above and below by these bounds, will be asymptotically equivalent as well. In order to calculate the rate of growth, we distinguish two different scenarios of interest. First, an scenario where we fix the number of antennas at one side of each link, for example $M$, and let the number of users, $K$, grow to infinity. 
Under this assumption, it is clear that the growth rate of \eqref{eq:rate_of_growth} will be dominated by the second addend, $K(M-1) \log \frac{K-M}{M-1}$, and, thus 
\begin{equation}
\log(\#(\pi_1^{-1}(H_0))) \in \Theta(K\log K),
\end{equation}
where $\Theta(K\log K)$ represents the class of functions that are asymptotically bounded both above and below by $K\log K$. Equivalently, $c_1 K\log K \leq \log(\#(\pi_1^{-1}(H_0))) \leq c_2 K\log K$ for some positive $c_1$ and $c_2$. Note that $\Theta(K\log K)$ denotes a polynomial rate of growth which is faster than linear, $\Theta(K)$, but slower than quadratic, $\Theta(K^2)$. Consequently, it can be said that the logarithm of the number of solutions grows as $K^{1+c}$ where $c\in (0,1)$, i.e. the number of solutions grows exponentially with $K^{1+c}$.

Now, we consider a second scenario where the ratio $\gamma=M/N$ is fixed. Given that $M+N=K+1$, we have that both $M$ and $N$ will grow as fast as $K$, i.e. $N=\frac{K+1}{\gamma+1}$ and $M=\frac{\gamma}{\gamma+1}(K+1)$. Taking this into account, it is trivial to see that both terms on the right hand side of \eqref{eq:rate_of_growth} grow as $K^2$ and, consequently 
\begin{equation}
\log(\#(\pi_1^{-1}(H_0))) \in \Theta(K^2).
\end{equation}
In summary, the logarithm of the number of solutions is quadratic in $K$ or, in other words, the number of solutions grows exponentially with $K^2$. Note that this rate is asymptotically equivalent to that obtained from B\'{e}zout's Theorem
which bounds the number of solutions by $2^{K(K-1)}$. Despite being asymptotically equivalent, the upper bound proposed herein is remarkably tighter.

\section{Numerical experiments}
\label{sec:numerical_experiments}
In this section we present some results obtained by means of the integral formulae in Theorem \ref{th:main2} (for arbitrary interference channels) and Theorem \ref{th:main3} (for square symmetric interference channels). We first evaluate the accuracy provided by the approximation of the integrals by Monte Carlo methods. To this end, we focus initially on single-beam systems, for which the procedure described in Section \ref{sec:single_beam} allows us to efficiently obtain the exact number of IA solutions for a given scenario. The true number of solutions can thus be used as a benchmark to assess the accuracy of the approximation.

Table \ref{tab:num_solutions_1beam_exact_approx} compares the number of solutions given by both the exact and the approximate procedures. To simplify the analysis, we have considered $(M \times (K-M+1),1)^K$ symmetric single-beam networks for increasing values of $M$ and $K$. As shown in Section \ref{sec:connection_graph_theory}, counting IA solutions for this scenario is equivalent to the well-studied graph theory problem of counting siple loop-free labeled $(M-1)$-regular digraphs with $K$ nodes. Thus, additional terms and further information can be retrieved from integer sequences databases such as \cite{OEIS} from its corresponding A-number given in the last row of Table \ref{tab:num_solutions_1beam_exact_approx}. Percentages represent the estimated relative error, $2\varepsilon\cdot 100$, for each scenario (see Section \ref{sec:montecarlo}).
\begin{DIFnomarkup}
\begin{table*}[h!]\centering
\renewcommand\arraystretch{1}
\begin{tabular}{@{}cccccc@{}}\toprule
& $M=2$ & \phantom{a}& $M=3$ & \phantom{a} & $M=4$\\
\cmidrule{2-2} \cmidrule{4-4} \cmidrule{6-6}
& $(2\times (K-1),1)^K$ &&  $(3\times (K-2),1)^K$ &&  $(4\times (K-3),1)^K$ \\
& Exact / Approx. &&  Exact / Approx. &&  Exact / Approx. \\ \midrule
  $K=2$ & 1 / 1 $\pm$ 0.0 \% && -- && -- \\
   $K=3$ & 2 / 2 $\pm$ 1.0 \% && 1 / 1 $\pm$ 0.5 \% && -- \\
   $K=4$ & 9 / 9 $\pm$ 1.6 \% && 9 / 9 $\pm$ 1.6 \% && 1 / 1 $\pm$ 0.6 \% \\
   $K=5$ & 44 / 44 $\pm$ 2.6 \% && 216 / 216 $\pm$ 1.5 \% && 44 / 44 $\pm$ 2.6 \% \\
   $K=6$ & 265 / 266 $\pm$ 3.3 \% && 7\,570 / 7\,291 $\pm$ 5.5 \% && 7\,570 / 7\,291 $\pm$ 5.5 \% \\
   $K=7$ & 1\,854 / 1\,868 $\pm$ 9.6 \% && 357\,435 / 361\,762 $\pm$ 8.7 \% && 1\,975\,560 / 1\,936\,679 $\pm$ 7.0 \% \\
   $K=8$ & 14\,833 / 13\,144 $\pm$ 20.6 \%&& 22\,040\,361 / 22\,419\,610 $\pm$ 11.3 \% && 749\,649\,145 / 739\,668\,504 $\pm$ 14.1 \% \\
   &  $\vdots$ && $\vdots$ &&  $\vdots$\\
   $K>8$ & \cite[Seq. A000166]{OEIS} && \cite[Seq. A007107]{OEIS} && \cite[Seq. A007105]{OEIS}\\
\bottomrule
\end{tabular}
\caption{Comparison of exact and approximate number of IA solutions for several symmetric single-beam scenarios, $(M \times (K-M+1),1)^K$.}
\label{tab:num_solutions_1beam_exact_approx}
\end{table*}
\end{DIFnomarkup}
Figure \ref{fig:asymptotic_growth_rate_fixed_M} depicts the evolution of the exact number of solutions with a growing $K$, and the area between the proposed upper and lower bounds, for different values of $M$ (form top to bottom, $M=2,3,4$). It shows that all three are asymptotically equivalent, as proved in Section \ref{sec:bounds_growth_rate}. The exact number of solutions has been obtained from the A-sequences mentioned in Table \ref{tab:num_solutions_1beam_exact_approx}. For the case $M=4$, the solid line corresponds to the values which are available at the time of writing in \cite[Seq. A007105]{OEIS}, i.e. $K\leq 14$. Beyond that point, the dashed line extrapolates new values following the model $aK\log(K)+bK+c$. The coefficients $a,b$ and $c$ are those providing the best least squares fit of the available data for $K\leq 14$.

\begin{figure}[h!]
\centering
\includegraphics[width=12cm]{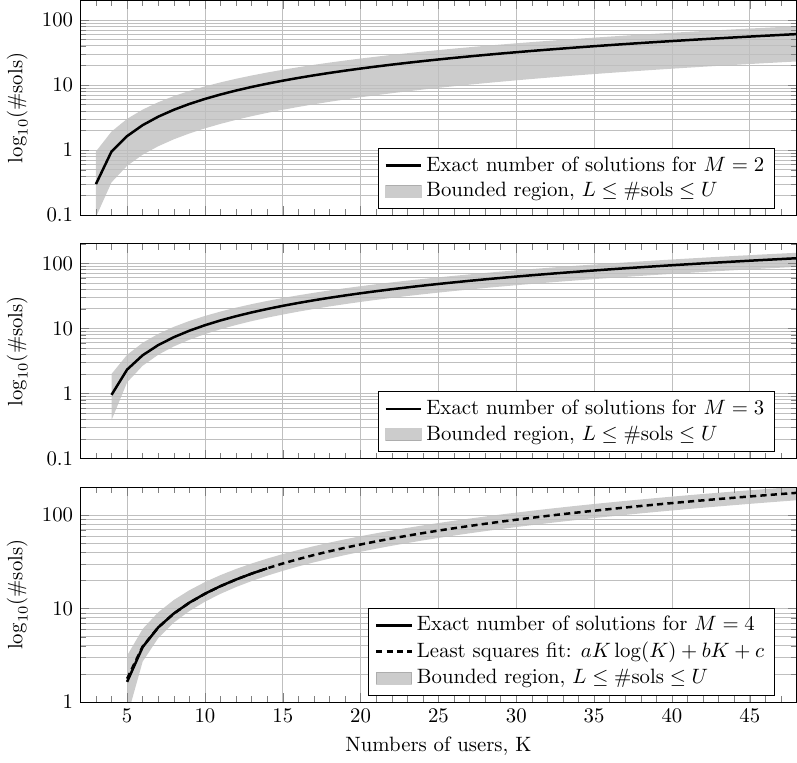}
\caption{Growth rate of the number of IA solutions in single beam systems, $(M \times (K-1),1)^K$, for $M=2,3,4$.}
\label{fig:asymptotic_growth_rate_fixed_M}
\end{figure}

Now we move to multi-beam scenarios, for which the exact number of solutions is only know for a few scenarios. Table \ref{tab:num_solutions_2beam_approx} shows the results obtained for some instances of the $(M \times (2K-M+2),2)^K$ network. These results have been obtained using the integral formula in Theorem \ref{th:main2}, except the square cases ($M=N$), for which we used the expression in Theorem \ref{th:main3}. For instance, we can mention that the system $(5 \times 5, 2)^4$ has, with a high confidence level, about 3700 different solutions (this result has been independently confirmed in \cite{Bresler2013}). As numerical results show, the integral formula in Theorem \ref{th:main3} can be approximated much faster than that of Theorem \ref{th:main2}, thus allowing us to get smaller relative errors. For the sake of completeness, Table \ref{tab:num_solutions_square} shows the approximate number of solutions for some additional square symmetric multi-beam scenarios. For some of them the exact number of solutions was already known, as indicated in the table. For others (those indicated as N/A in the table) the exact number of solutions was unknown.
\begin{DIFnomarkup}
\begin{table*}[h!]\centering
\renewcommand\arraystretch{1}
\begin{tabular}{@{}cccccccc@{}}\toprule
& $M=3$ & \phantom{a}& $M=4$ & \phantom{a} & $M=5$ & \phantom{a} & $M=6$\\
\cmidrule{2-2} \cmidrule{4-4} \cmidrule{6-6} \cmidrule{8-8}
& $(3\times (2K-1),2)^K$ &&  $(4\times (2K-2),2)^K$ &&  $(5\times (2K-3),2)^K$ &&  $(6\times (2K-4),2)^K$ \\ \midrule
   $K=2$ & 0 $\pm$ 0.0 \% && 1 $\pm$ 4.1 \% && -- && -- \\
   $K=3$ & 1 $\pm$ 4.2 \% && 6 $\pm$ 0.0 \% && 1 $\pm$ 4.8 \% && 1 $\pm$ 5.2 \% \\
   $K=4$ & 9 $\pm$ 5.8 \% && 973 $\pm$ 7.0 \% && 3\,700 $\pm$ 0.1 \% && 973 $\pm$ 7.0 \% \\
   $K=5$ & 223 $\pm$ 14.8 \% && 530\,725 $\pm$ 11.3 \% && 72\,581\,239 $\pm$ 17.8 \% && 387\,682\,648 $\pm$ 0.7 \% \\
\bottomrule
\end{tabular}
\caption{Approximate number of IA solutions for several symmetric $2$-beam scenarios, $(M \times (2K-M+2),2)^K$.}
\label{tab:num_solutions_2beam_approx}
\end{table*}
\end{DIFnomarkup}

\begin{DIFnomarkup}
\begin{table*}[h!]\centering
\renewcommand\arraystretch{1}
\begin{tabular}{@{}cccccccccc@{}}\toprule
K & \phantom{a}& d & \phantom{a} & Scenario & \phantom{a} & Exact & Ref. &  \phantom{a} & Approximate \\ \midrule
3 && 1 && $(2\times 2,1)^3$ && 2 & \cite{Jafar08} && 2 $\pm$ 0.9 \% \\
3 && 2 && $(4\times 4,2)^3$ && 6 & \cite{Jafar08} && 6 $\pm$ 0.9 \% \\
3 && 3 && $(6\times 6,3)^3$ && 20 & \cite{Jafar08} && 20 $\pm$ 1.4 \% \\
4 && 2 && $(5\times 5,2)^4$ && N/A & N/A && 3\,700 $\pm$ 0.1 \% \\
4 && 4 && $(10\times 10,4)^4$ && N/A & N/A && 13\,887\,464\,893\,004 $\pm$ 6.8 \% \\
5 && 1 && $(3\times 3,1)^5$ && 216 & \cite{Schmidt2010} && 216 $\pm$ 0.6 \% \\
5 && 2 && $(6\times 6,2)^5$ && N/A & N/A && 387\,724\,347 $\pm$ 0.7 \% \\
\bottomrule
\end{tabular}
\caption{Approximate number of IA solutions for selected square symmetric scenarios, \smash{$(\frac{K+1}{2}d \times \frac{K+1}{2}d,d)^K$}.}
\label{tab:num_solutions_square}
\end{table*}
\end{DIFnomarkup}

Although these results have a mainly theoretical interest, they might also have some important practical implications. In the following, we illustrate this point with a numerical experiment. Let us assume that we have a moderate-size network for which the total number of solutions is relatively small. One such example could be the $(4 \times 6, 2)^4$ system which, according to the results in Table \ref{tab:num_solutions_2beam_approx}, has a number solutions in the interval $[904,1042]$ ($95\%$ confidence interval).
It is obvious that, since the exact number of solutions is unknown, a systematic way to compute all interference alignment solutions for a given channel realization does not exist. Still, one may try to compute them by repeatedly running some iterative algorithm such as the ones in \cite{Gomadam11}, \cite{Gonzalez2014a} or \cite{OscarIcassp} from different initialization points.\footnote{Note that the last one is restricted to single-beam scenarios and, consequently, cannot be applied to the scenario at hand.} This idea is illustrated in Figure \ref{fig:sumrate_comparison}, where the sum-rate performance associated to 973 different solutions is shown. The fact that we have been able to find 973 solutions only demonstrates that, at least, 973 solutions exist. The actual number may be even larger and presumably below 1042, but it seems hard to be determined by means of the algorithm in \cite{Gomadam11}.
In Figure \ref{fig:sumrate_comparison}, the maximum sum-rate solution is represented with a thicker solid line, while the average sum-rate of all solutions is represented with a dashed line. Interestingly, the relative performance improvement provided by the maximum sum-rate solution over the average is substantial, i.e. it is always above 10 \% for SNR values below 40 dB, and more than 20 \% for SNR=20 dB. We note that this improvement is comparable to the one provided by sum-rate optimization algorithms which take into account additional information in the optimization procedure such as direct channels and noise variance.

\begin{figure}[h!]
\centering
\includegraphics[width=12cm]{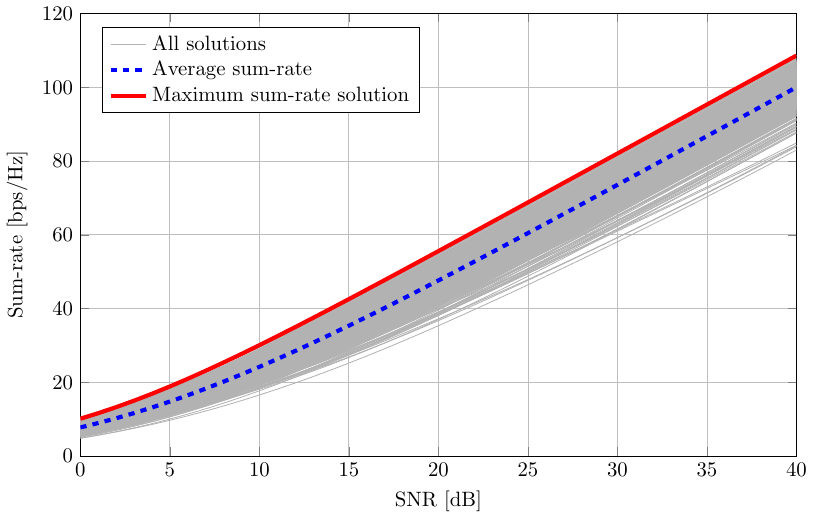}
\caption{Comparison of the sum rate achieved by 973 different solutions for the system $(4\times 6,2)^4$.}
\label{fig:sumrate_comparison}
\end{figure}


\section{Conclusion}
In this paper we have provided two integral formulae to compute the finite number of IA solutions in MIMO interference channels, including multi-beam ($d_k > 1$) systems. The first one can be applied to arbitrary $K$-user interference channels, whereas the second one solves the symmetric square case. Both integrals can be estimated by means of Monte Carlo methods. We have also specialized our results to single-beam networks, leading to a combinatorial counting procedure that allows to obtain the exact number of solutions and interesting connections with well-known graph counting problems.

\appendices

\section{Mathematical preliminaries}

\label{appendix:preliminaries}

To facilitate reading, in this section we recall the mathematical results used in this paper. Firstly, we provide a short review on mappings between Riemannian manifolds and the main mathematical result used to derive the number of IA solutions, which is the Coarea formula. Secondly, we review the volume of the Grassmanian manifolds and the volume of the unitary group, which are also used throughout the paper.

\subsection{Tubes in Riemannian manifolds and the Coarea formula}
A general result about tubes states that the volume of a tubular neighborhood about a compact embedded submanifold is essentially given by the intrinsic volume of the submanifold times the volume of a ball of the appropiate dimension. We write down a simplified version of \cite[Th. 9.23]{Gray04}:
\begin{theorem}\label{th:gray}
Let $X$ be a compact, embedded, (real) codimension $c$ submanifold of the Riemannian manifold $Y$. Then, for sufficiently small $\epsilon>0$,
\[
Vol(y\in Y:d(y,X)<\epsilon)=Vol(X)Vol(r\in\R^c:\|r\|\leq1)\epsilon^c+O(\epsilon^{c+1}).
\]
Here, $Vol(X)$ is the volume of $X$ w.r.t. its natural Riemannian structure inherited from that of $Y$.
\end{theorem}
 One of our main tools is the so--called Coarea Formula. The most general version we know may be found in \cite{Fe69}, but for our purposes a smooth version as used in \cite[p. 241]{BlCuShSm98} or \cite{Ho93} suffices. We first need a definition.
\begin{definition}\label{def:NJ}
Let $X$ and $Y$ be Riemannian manifolds, and let $\varphi:X\longrightarrow Y$ be a $C^1$ surjective map. Let $k=\dim(Y)$ be the real dimension of $Y$. For every point $x\in X$ such that the differential mapping $D\varphi(x)$ is surjective,  let $v_1^x,\ldots,v_k^x$ be an orthogonal basis of $Ker(D\varphi(x))^\perp$. Then, we define the Normal Jacobian of $\varphi$ at $x$, ${\rm NJ}\varphi(x)$, as the volume in the tangent space $T_{\varphi(x)}Y$ of the parallelepiped spanned by $D\varphi(x)(v_1^x),\ldots,D\varphi(x)(v_k^x)$. In the case that $D\varphi(x)$ is not surjective, we define ${\rm NJ}\varphi(x)=0$.
\end{definition}
\begin{theorem}[Coarea formula]\label{th:coarea}
Let $X,Y$ be two Riemannian manifolds of respective dimensions $k_1\geq k_2$. Let $\varphi:X\longrightarrow Y$ be a $C^\infty$ surjective map, such that the differential mapping $D\varphi(x)$ is surjective for almost all $x\in X$. Let $\psi:X\longrightarrow\R$ be an integrable mapping. Then, the following equality holds:
\begin{equation}\label{eq:coarea}
\int_{x\in X}\psi(x) {\rm NJ}\varphi(x)\;dX=\int_{y\in Y}\int_{x\in \varphi^{-1}(y)}\psi(x)~dx~dy.
\end{equation}
\end{theorem}
Note that from the Preimage Theorem and Sard's Theorem (see \cite[Ch. 1]{GuilleminPollack1974}), the set $\varphi^{-1}(y)$ is a manifold of dimension equal to $\dim(X)-\dim(Y)$ for almost every $y\in Y$. Thus, the inner integral of (\ref{eq:coarea}) is well defined as an integral in a manifold. Moreover, if $\dim(X)=\dim(Y)$ then $\varphi^{-1}(y)$ is a finite set for almost every $y$, and then the inner integral is just a sum with $x\in\varphi^{-1}(y)$.

The following result, which follows from the Coarea formula, is \cite[p. 243, Th. 5]{BlCuShSm98}.
\begin{theorem}\label{th:doublefibration}
Let $X,Y$ and $\mathcal{V}\subseteq X\times Y$ be smooth Riemannian manifolds, with $\dim(\mathcal{V})=\dim(X)$ and $\mathcal{Y}$ compact. Assume that $\pi_2:\mathcal{V}\rightarrow Y$ is regular (i.e. $D\pi_2$ is everywhere surjective) and that $D\pi_1(x,y)$ is surjective for every $(x,y)\in \mathcal{V}$ out of some zero measure set. Then, for every open set $U\subseteq\mathcal{X}$ contained in some compact set $K\subseteq\mathcal{X}$,
\begin{equation}\label{eq:doublef}
\int_{x\in U}\#(\pi_1^{-1}(x))\,dx=\int_{y\in Y}\int_{x\in U:(x,y)\in \CV}DET(x,y)^{-1}\,dx\,dy,
\end{equation}
where $DET(x,y)=\det(DG_{x,y}(x)DG_{x,y}(x)^*)$ and $G_{x,y}$ is the (locally defined) implicit function of $\pi_1$ near $x=\pi_1(x,y)$. That is, close to $(x,y)$ the sets $\CV$ and $\{(x,G_{x,y}(x))\}$ coincide.
\end{theorem}

\begin{corollary}\label{cor:doublefibration}
In addition to the hypotheses of Theorem \ref{th:doublefibration}, assume that there exists $y_0\in Y$ such that for every $y\in Y$ there exists an isometry $\varphi_y:Y\rightarrow Y$ with $\varphi_y(y)=y_0$ and an associated isometry $\chi_y:X\rightarrow X$ such that $\chi_y(U)=U$ and $(\chi_y\times\varphi_y)(\CV)=\CV$. Then,
\[
\int_{x\in U}\#(\pi_1^{-1}(x))\,dx=Vol(Y)\int_{x\in U:(x,y_0)\in \CV}DET(x,y_0)^{-1}\,dx.
\]
\end{corollary}
\begin{proof}
Let $y\in Y$ and let $\varphi_y,\chi_y$ as in the hypotheses. Then, consider the mapping
\[
\begin{matrix}
\chi_y\mid_{\{x\in U:(x,y)\in\CV\}}:&\{x\in U:(x,y)\in\CV\}:&\rightarrow&\{x\in U:(x,y_0)\in\CV\}\\
&x&\mapsto&\chi_y(x),
\end{matrix}
\]
which is the restriction of an isometry, hence an isometry. Let $G_{x,y}$ be the local inverse of $\pi_1$ close to $(x,y)\in \CV$. The change of variables formula then implies:
\begin{equation}\label{eq:cov}
\int_{x\in U:(x,y)\in\CV}DET(x,y)^{-1}\,dx=
\int_{x\in U:(x,y_0)\in\CV}DET(\chi_y^{-1}(x),y)^{-1}\,dx.
\end{equation}
 Note that the following diagram is commutative:
\[
\begin{matrix}
\CV\cap\pi_1^{-1}(U)&\overset{\chi_y^{-1}\times\varphi_y^{-1}}{\longrightarrow}&\CV\cap\pi_1^{-1}(U)\\
\pi_1\downarrow\uparrow G_{x,y_0}&&\downarrow\pi_1\\
X&\overset{\chi_y^{-1}}{\longrightarrow}&X
\end{matrix}
\]
Thus, the mapping $(\chi_y^{-1}\times\varphi_y^{-1})\circ G_{x,y_0}\circ\chi_y$ is a local inverse of $\pi_1$ near $(\chi_y^{-1}(x),y)$, that is
\[
G_{\chi_y^{-1}(x),y}=(\chi_y^{-1}\times\varphi_y^{-1})\circ G_{x,y_0}\circ\chi_y,
\]
and the composition rule for the derivative gives:
\[
DG_{\chi_y^{-1}(x),y}(\chi_y^{-1}(x))=D(\chi_y^{-1}\times\varphi_y^{-1})(G_{x,y_0}(x)) DG_{x,y_0}(x)D\chi_y^{-1}(\chi_y(x)).
\]
Now, $\chi_y$, $\varphi_y$ and $\chi_y\times\varphi_y$ are isometries of their respective spaces. Thus, we conclude:
\[
\det(DG_{\chi_y^{-1}(x)}(\chi_y^{-1}(x))DG_{\chi_y^{-1}(x),y}(\chi_y^{-1}(x))^*)=\det(DG_{x,y_0}(x)DG_{x,y_0}(x)^*),
\]
that is $DET(\chi_y^{-1}(x),y)=DET(x,y_0)$. Then, (\ref{eq:cov}) reads
\[
\int_{x\in U:(x,y)\in\CV}DET(x,y)^{-1}\,dx=
\int_{x\in U:(x,y_0)\in\CV}DET(x,y_0)^{-1}\,dx.
\]
That is, the inner integral in the right--hand side term (\ref{eq:doublef}) is constant. The corollary follows.
\end{proof}

\subsection{The volume of classical spaces}
Some helpful formulas are collected here:
\begin{equation}\label{eq:volumesphere}
\text{(cf. \cite[p. 248]{Gray04})}\quad Vol(\S(\C^a))=Vol(\S(\R^{2a}))=\frac{2\pi^{a}}{\Gamma(a)}
\end{equation}
is the volume of the complex sphere of dimension $a$.
\begin{equation}\label{eq:volumeunitary}
\text{(cf. \cite[p. 54]{Hua1979})}\quad Vol(\CU_{a})=\frac{(2\pi)^{\frac{a(a+1)}{2}}}{\Gamma(1)\cdots\Gamma(a)},
\end{equation}
is the volume of the unitary group of dimension $a$. Note that, as pointed out in \cite[p. 55]{Hua1979} there are other conventions for the volume of unitary groups. Our choice here is the only one possible for Theorem \ref{th:gray} to hold: the volume of $\CU_a$ is the one corresponding to its Riemannian metric inherited from the natural Frobenius metric in $\mathcal{M}_a(\C)$.

We finally recall the volume of the complex Grassmannian. Let $1\leq a\leq b$; then,
\begin{equation}\label{eq:volumegrass}
Vol(\G{a}{b})=\pi^{a(b-a)}\frac{\Gamma(2)\cdots \Gamma(a)\cdot\Gamma(2)\cdots\Gamma(b-a)}{\Gamma(2)\cdots\Gamma(b)}.
\end{equation}


\section{Proof of Theorem \ref{th:main1}}\label{sec:proof1}
We will apply Corollary \ref{cor:doublefibration} to the double fibration given by (\ref{eq:diag}). In the notations of Corollary \ref{cor:doublefibration}, we consider $X=\CH$, $Y=\CS$, $\mathcal{V}$ the solution variety and
\[
y_0=\left(\binom{I_{d_k}}{0_{N_k-d_k}},\binom{I_{d_k}}{0_{M_k-d_k}}\right)=(U_0,V_0)\in\CS.
\]
Given any other element $y=(U_k,V_k)\in\CS$, let $P_k$ and $Q_k$ be unitary matrices of respective sizes $N_k$ and $M_k$ such that
\[
U_k=P_k\binom{I_{d_k}}{0_{N_k-d_k}},\quad V_k=Q_k\binom{I_{d_k}}{0_{M_k-d_k}}.
\]
Then consider the mapping
\[
\varphi_y(\tilde{U}_k,\tilde{V}_k)=(P_k^*\tilde{U}_k,Q_k^*\tilde{V}_k),
\]
which is an isometry of $\CS$ and satisfies $\varphi_y(y)=y_0$ as demanded by Corollary \ref{cor:doublefibration}. We moreover have the associated mapping $\chi_y:\CH\rightarrow\CH$ given by
\[
\chi_y((H_{kl})_{k\neq l})=(P_k^TH_{kl}Q_l)_{k\neq l}
\]
which is an isometry of $\CH$. Moreover, $\chi_y(\CH_\epsilon)=\CH_\epsilon$ and $\chi_y\times\varphi_y(\CV)=\CV$. We can thus apply Corollary \ref{cor:doublefibration} which yields
\begin{equation}\label{eq:doubleftext}
\int_{H\in \CH_\epsilon}\#(\pi_1^{-1}(x))\,dx=vol(\CS)\int_{H\in\CH_I\cap\CH_\epsilon}\det(DG(H)DG(H)^*)^{-1}\,dH,
\end{equation}
where $G$ is the local inverse of $\pi_1$ close to $H$ at $(H,U_0,V_0)$. We now compute $\det(DG(H)DG(H)^*)^{-1}$. From the definition of $G$ we have
\[
T_{(H,U,V)}\CV=\{(\dot H,DG(H)\dot H):\dot H\in T_H\CH\}.
\]
On the other hand, from the defining equations (\ref{eq:1}) and considering $H\in\mathcal{H}_I$ and $\dot H\in T_H\CH$ as block matrices
\[
H=\begin{pmatrix}0_{d_k\times d_l}&A\\B&C\end{pmatrix},\quad \dot H=\begin{pmatrix}\dot R_{kl}&\dot A_{kl}\\\dot B_{kl}&\dot C_{kl}\end{pmatrix}
\]
we can identify
\[
T_{(H,y_0)}\CV=\left\{\left(\dot H,\binom{0}{\dot U},\binom{0}{\dot V}\right):\dot U_k^TB_{kl}+\dot R_{kl}+A_{kl}\dot V_l=0,k\neq l\right\}=
\]
\[
\{(\dot H,\dot U,\dot V):(\dot U,\dot V)=
-\Psi_{H}^{-1}( \dot R_{kl})\}.
\]
Hence\footnote{Note that in the appendices we will sometimes refer to $\Psi$ as $\Psi_H$ to make the dependence on $H$ explicit.},
\[
DG(H)\dot H=-\Psi_{H}^{-1}( \dot R_{kl})=-\Psi_H^{-1}(U_0^*\dot H_{kl} V_0).
\]
A straight--forward computation shows that:
\[
DG(H)^*(\dot U,\dot V)=(-U_0\Psi_{H}^{-*}(\dot U,\dot V)V_0^*)_{k\neq l}.
\]
Thus, writing $\Psi=\Psi_{H}$, we have:
\[
DG(H)DG(H)^*(\dot U,\dot V)=\Psi^{-1}\Psi^{-*}(\dot U,\dot V).
\]
Therefore, $(DG(H)DG(H)^*)^{-1}=\Psi^*\Psi$ and
\[
\det(DG(H)DG(H)^*)^{-1}=\det(\Psi^*\Psi)=|\det(\Psi)|^2=\det(\Psi\Psi^*).
\]
From this last equality and (\ref{eq:doubleftext}) we have:
\[
\int_{H\in \CH_\epsilon}\#(\pi_1^{-1}(H))\,dH=Vol(\CS)\int_{H\in\CH_I\cap\CH_\epsilon}\det(\Psi\Psi^*)\,dH.
\]
Theorem \ref{th:main1} follows dividing both sides of this equation by $Vol(\CH_\epsilon)$ and using the fact that for every choice of $H$ out of a zero measure set, the number of elements in $\pi_1^{-1}(H)$ is constant (see Lemma \ref{lem:constant}).


\section{Proof of Theorem \ref{th:main2}}\label{sec:proof2}
Let $\epsilon<1$ and let $\CH_\epsilon$ be the product for $k\neq l$ of the sets
\[
\{H_{kl}:d(H_{kl},\{R\in\mathcal{M}_{N_k\times M_l}(\C):\|R\|_F=1\})<\epsilon\}.
\]
From Theorem \ref{th:gray}, each of these sets have volume equal to
\[
2Vol(\{R\in\mathcal{M}_{N_k\times M_l}(\C):\|R\|_F=1\})\epsilon+O(\epsilon^{2})\underset{(\ref{eq:volumesphere})}{=}\frac{4\pi^{N_kM_l}}{\Gamma(N_kM_l)}\epsilon+O(\epsilon^2)
\]
Thus,
\[
Vol(\CH_\epsilon)=\left(\prod_{k\neq l}\frac{4\pi^{N_kM_l}}{\Gamma(N_kM_l)}\right)\epsilon^{K(K-1)}+O(\epsilon^{K(K-1)+1}).
\]
On the other hand, consider the smooth mapping
\[
\begin{matrix}
f:&\CHI&\rightarrow&\CHI\cap\prod_{k\neq l}\{H_{kl}:\|H_{kl}\|_F=1\}\\
&\left(H_{kl}\right)_{k,l}&\rightarrow&\left(\frac{H_{kl}}{\|H_{kl}\|_F}\right)_{k,l}
\end{matrix}
\]
and apply Theorem \ref{th:coarea} to get
\[
\int_{H\in\CH_I\cap\CH_\epsilon}\det(\Psi_H\Psi_H^*)\,dH= \int_{H\in \CHI\cap\prod_{k\neq l}\{H_{kl}:\|H_{kl}\|_F=1\}}\int_{\vec{t}=(t_{kl})\in[-\epsilon,\epsilon]^{K(K-1)}}\det(\Psi_{\hat{H}} \Psi_{\hat{H}}^*){\rm NJ} f(\hat{H})\,d\vec{t}\,dH,
\]
where $\hat{H}_{kl}=H_{kl}(1+t_{kl})$. Note that the function inside the inner integral is smooth and hence for any $H\in \CHI\cap\prod_{k\neq l}\{H_{kl}:\|H_{kl}\|_F=1\}$ we have
\[
\det(\Psi_{\hat{H}} \Psi_{\hat{H}}^*){\rm NJ} f(\hat{H})= \det(\Psi_{H} \Psi_{H}^*){\rm NJ} f(H)+O(\epsilon).
\]
We have thus proved (using $\approx$ for equalities up to $O(\epsilon)$):
\[
\int_{H\in\CH_I\cap\CH_\epsilon}\det(\Psi_H\Psi_H^*)\,dH\approx \int_{H\in \CHI\cap\prod_{k\neq l}\{H_{kl}:\|H_{kl}\|_F=1\}}(2\epsilon)^{K(K-1)} \det(\Psi_{H} \Psi_{H}^*){\rm NJ} f(H)\,dH,
\]
It is very easy to see that ${\rm NJ} f(H)=1$ if $H=(H_{kl})$ with $\|H_{kl}\|_F=1$. Thus, we have
\[
\int_{H\in\CH_I\cap\CH_\epsilon}\det(\Psi_H\Psi_H^*)\,dH\approx (2\epsilon)^{K(K-1)}\int_{H\in \CHI\cap\prod_{k\neq l}\{H_{kl}:\|H_{kl}\|_F=1\}} \det(\Psi_{H} \Psi_{H}^*)\,dH.
\]
From Theorem \ref{th:main1} and taking limits we then have that for almost every $H_0\in\CH$,
\begin{equation}\label{eq:xxx}
\#(\pi_1^{-1}(H_0))=C
\dashint_{H\in \CHI\cap\prod_{k\neq l}\{H_{kl}:\|H_{kl}\|_F=1\}}\det(\Psi\Psi^*)\,dH,
\end{equation}
where
\[
C=\frac{2^{K(K-1)}Vol\left(H\in \CHI\cap\prod_{k\neq l}\{H_{kl}:\|H_{kl}\|_F=1\}\right)}{\prod_{k\neq l}\frac{4\pi^{N_kM_l}}{\Gamma(N_kM_l)}}Vol(\CS)
\]
Now, $\CHI\cap\prod_{k\neq l}\{H_{kl}:\|H_{kl}\|_F=1\}$ is a product of spheres and thus from (\ref{eq:volumesphere})
\[
Vol\left(\CHI\cap\prod_{k\neq l}\{H_{kl}:\|H_{kl}\|_F=1\}\right)=\prod_{k\neq l}\frac{2\pi^{N_kM_l-d_kd_l}}{\Gamma(N_kM_l-d_kd_l)}.
\]
Finally, $\CS=\left(\prod_{k}\G{d_k}{N_k}\right)\times \left(\prod_{l}\G{d_l}{M_l}\right)$ is a product of complex Grassmannians, and its volume is thus the product of the respective volumes, given in (\ref{eq:volumegrass}).  That is,
\[
Vol(\CS)=\left(\prod_{k}\pi^{d_k(N_k-d_k)}\frac{\Gamma(2)\cdots \Gamma(d_k)\cdot\Gamma(2)\cdots\Gamma(N_k-d_k)}{\Gamma(2)\cdots\Gamma(N_k)} \right)\times
\]
\[
\left(\prod_{l}\pi^{d_l(M_l-d_l)}\frac{\Gamma(2)\cdots \Gamma(d_l)\cdot\Gamma(2)\cdots\Gamma(M_l-d_l)}{\Gamma(2)\cdots\Gamma(M_l)}\right).
\]
Putting these computations together, and using $s=0$, we get the value of $C$ claimed in Theorem \ref{th:main2}.


\section{Proof of Theorem \ref{th:main3}}
\label{sec:proof3}

The proof of this theorem is quite long and nontrivial. We will apply Theorem \ref{th:main1} to the sets
\begin{equation}\label{eq:Hepsilon}
\CH_\epsilon=\{(H_{kl}):d(H_{kl},\CU_{N_k})\leq\epsilon,k\neq l\}.
\end{equation}
Then, because (\ref{eq:main1}) holds for every $\epsilon$, one can take limits and conclude that for almost every $H_0\in\CH$,
\begin{equation}\label{eq:Hepsilon2}
\#(\pi_1^{-1}(H_0))=\lim_{\epsilon\rightarrow0}\frac{Vol(\CHI\cap\CH_\epsilon)Vol(\CS)}{Vol(\CH_\epsilon)}\dashint_{H\in \CHI\cap\CH_\epsilon}\det(\Psi\Psi^*)\,dH.
\end{equation}
The claim of Theorem \ref{th:main3} will follow from the (difficult) computation of that limit. We organize the proof in several subsections.

\subsection{Unitary matrices with some zeros}
In this section we study the set of unitary matrices of size $N\geq2d$ which have a principal $d\times d$ submatrix equal to $0$, and the set of closeby matrices. For simplicity of the exposition, the notations of this section are inspired in, but different from, the notations of the rest of the paper. Let
\[
\mathcal{T}=\mathcal{T}_{N,d}=\left\{H=\begin{pmatrix}0_{d\times d}&A\\B&C\end{pmatrix}\right\}\subseteq\mathcal{M}_{N\times N}(\C).
\]
Note that $\mathcal{T}$ is a vector space of complex dimension $N^2-d^2$. Our three main results are:
\begin{proposition}\label{prop:T}
The set $\CU_N\cap\mathcal{T}$ is a manifold of codimension $N^2$ inside $\mathcal{T}$. Moreover,
\[
Vol(\CU_N\cap\mathcal{T})=\frac{Vol(\CU_{N-d})^2}{Vol(\CU_{N-2d})}.
\]
\end{proposition}

\begin{proposition}\label{prop:comparevolumes}
The following equality holds:
\[
\lim_{\epsilon\rightarrow0}\frac{Vol(H\in\mathcal{T}:d(H,\CU_N)\leq\epsilon)}{\epsilon^{N^2}}=2^{d^2}Vol(\CU_N\cap\mathcal{T})Vol(x\in\R^{N^2}:\|x\|\leq1).
\]
\end{proposition}
\begin{proposition}\label{prop:limits}
Let $\psi:\mathcal{T}\rightarrow\R$ be a smooth mapping defined on $\mathcal{T}$ and such that $\psi(H)$ depends only on the $A$ and $B$ part of $H$, but not on the part $C$. Denote $\psi(H)=\psi(A,B)$. Then,
\[
\lim_{\epsilon\rightarrow0}\frac{\int_{H\in\mathcal{T}:d(H,\CU_N)\leq \epsilon}\psi(H)\,dH}{Vol(H\in\mathcal{M}_N(\C):d(H,\CU_N)\leq\epsilon)}=
\]
\[
 \frac{2^{d^2}Vol(\CU_{N-d})^2}{Vol(\CU_N)Vol(\CU_{N-2d})}\dashint_{(A^*, B) \in \CU_{(N-d)\times d}}\psi(A,B) \,d(A,B).
\]
\end{proposition}

\subsubsection{Proof of Proposition \ref{prop:T}}\label{sec:proofpropT}
Let
\begin{equation}\label{eq:xi}
\begin{matrix}
\xi:&\CU_{N-d}^2&\rightarrow&\CU_N\cap\mathcal{T}\\
&(U,V)&\mapsto&\begin{pmatrix}I_d&0\\0&U\end{pmatrix} J \begin{pmatrix}I_d&0\\0&V^*\end{pmatrix}
\end{matrix}
\end{equation}
where
\[
J=\begin{pmatrix}0&I_d&0\\I_d&0&0\\0&0&I_{N-2d}\end{pmatrix}.
\]
We claim that $\xi$ is surjective. Indeed, let
\[
H=\begin{pmatrix}0&A\\B&C\end{pmatrix}\in\CU_N\cap\mathcal{T}.
\]
From $HH^*=I_N$ we have that $A$ satisfies $AA^*=I_d$, i.e. the rows of $A$ can be completed to form a unitary basis of $\C^{N-d}$. Namely, there exists $V\in\CU_{N-d}$ such that $A=(I_d\;0)V$. Similarly, there exists $U\in\CU_{N-d}$ such that $B=U\binom{I_d}{0}$. Then,
\[
H=\begin{pmatrix}I_d&0\\0&U\end{pmatrix} \begin{pmatrix}0&I_d&0\\I_d&R_1&
R_2\\0&R_3&R_4\end{pmatrix} \begin{pmatrix}I_d&0\\0&V\end{pmatrix},
\]
where
\[
R=\begin{pmatrix}R_1&R_2\\R_3&R_4\end{pmatrix}
\]
 satisfies $URV=C$. Now, this implies that the matrix
\[
\begin{pmatrix}0&I_d&0\\I_d&R_1&
R_2\\0&R_3&R_4\end{pmatrix}
\]
is unitary, which forces $R_1=0$, $R_2=0$, $R_3=0$ and $R_4$ unitary. That is
\[
H=\begin{pmatrix}I_d&0\\0&U\end{pmatrix} \begin{pmatrix}0&I_d&0\\I_d&0&0\\0&0&R_4\end{pmatrix} \begin{pmatrix}I_d&0\\0&V\end{pmatrix}=\begin{pmatrix}I_d&0\\0&U\end{pmatrix} J \begin{pmatrix}I_d&0&0\\0&I_d&0\\0&0&R_4\end{pmatrix} \begin{pmatrix}I_d&0\\0&V\end{pmatrix},
\]
that is
\[
H=\xi\left(U,V^*\begin{pmatrix}I_d&0\\0&R_4\end{pmatrix}^*\right),
\]
and the surjectivity of $\xi$ is proved. Moreover, this construction describes $\CU_N\cap\mathcal{T}$ as the orbit of $J$ under the action in $\mathcal{T}$ given by
\[
((U,V),X)\mapsto \begin{pmatrix}I_d&0\\0&U\end{pmatrix} X \begin{pmatrix}I_d&0\\0&V^*\end{pmatrix}.
\]
Then, $\CU_N\cap\mathcal{T}$ is a smooth manifold diffeomorphic to the quotient space
\[
\CU_{N-d}^2/S_J,
\]
where $S_J$ is the stabilizer of $J$. Now, $(U,V)\in S_J$ if and only if
\[
\begin{pmatrix}I_d&0&0\\0&U_1&U_2\\0&U_3&U_4\end{pmatrix} \begin{pmatrix}0&I_d&0\\I_d&0&0\\0&0&I_{N-2d}\end{pmatrix} \begin{pmatrix}I_d&0&0\\0&V_1^*&V_3^*\\0&V_2^*&V_4^*\end{pmatrix}= \begin{pmatrix}0&I_d&0\\I_d&0&0\\0&0&I_{N-2d}\end{pmatrix},
\]
which implies $U_1=I_d$, $U_2=0$, $U_3=0$, $V_1=I_d$, $V_2=0$, $V_3=0$ and $U_4=V_4$. Thus,
\begin{equation}\label{eq:S}
S_J=\left\{\left(\begin{pmatrix}I_d&0\\0&U_4\end{pmatrix},\begin{pmatrix}I_d&0\\0&U_4\end{pmatrix}\right):U_4\in\CU_{N-2d}\right\}.
\end{equation}
Then,
\[
\dim(\CU_N\cap\mathcal{T})=\dim(\CU_{N-d}^2/S_J)=2\dim(\CU_{N-d})^2-\dim(S_J)=2(N-d)^2-(N-2d)^2=N^2-2d^2.
\]
On the other hand, $\dim(\mathcal{T})=2N^2-2d^2$ and thus
\[
\text{codim}_\mathcal{T}(\CU_N\cap\mathcal{T})=2N^2-2d^2-(N^2-2d^2)=N^2,
\]
as claimed. We now apply the Coarea formula to $\xi$ to compute the volume of $\CU_N\cap\mathcal{T}$. Note that by unitary invariance the Normal Jacobian of $\xi$ is constant, and so is $Vol(\xi^{-1}(H))$. We can easily compute
\[
Vol(\xi^{-1}(H))\underset{\forall\;H}{=}Vol(\xi^{-1}(J))=Vol(S_J)\underset{(\ref{eq:S})}{=}\sqrt{2}^{(N-2d)^2}Vol(\CU_{N-2d}).
\]
For the Normal Jacobian of $\xi$, writing
\[
\dot U=\begin{pmatrix}\dot U_1&\dot U_2\\\dot U_3&\dot U_4\end{pmatrix},
\]
for an element in the tangent space to $\CU_{N-d}$ at $I_{N-d}$ (and similarly for $\dot V$), note that
\[
D\xi(I_{N-d},I_{N-d})(\dot U,\dot V)=\begin{pmatrix}0&\dot V_1^*&-\dot V_2\\\dot U_1&0&\dot U_2\\-\dot U_2^*&\dot V_2^*&\dot U_4+\dot V_4^*\end{pmatrix}.
\]
Thus, $D\xi(I_{N-d},I_{N-d})$ preserves the orthogonality of the natural basis of $T_U\CU_{N-d}\times T_V\CU_{N-d}$ but for the elements such that $\dot U_4\neq0$ or $\dot V_4\neq0$. We then conclude that $NJ(\xi)(I_{N-d},I_{N-d})=NJ(\eta)$ where
\[
\begin{matrix}
\eta:&\{M\in\mathcal{M}_{N-2d}(\C):M+M^*=0\}^2&\rightarrow&\{M\in\mathcal{M}_{N-2d}(\C):M+M^*=0\}\\
&(\dot U_4,\dot V_4)&\mapsto&\dot U_4+\dot V_4^*.
\end{matrix}
\]
It is a routine task to see that $\eta^*(L)=(L,L^*)$ which implies $\eta\eta^*(L)=2L$, that is
\[
\det(\eta\eta^*)=2^{\dim(\{M\in\mathcal{M}_{N-2d}(\C):M+M^*=0\})}=2^{(N-2d)^2}.
\]
Hence, $NJ(\eta)=\sqrt{\det(\eta\eta^*)}=\sqrt{2}^{(N-2d)^2}$. As we have pointed out above, the value of the Normal Jacobian of $\xi$ is constant. Thus, for every $U,V$,
\[
NJ(\xi)(U,V)=NJ(\eta)=\sqrt{2}^{(N-2d)^2}.
\]
The Coarea formula applied to $\xi$ then yields:
\[
Vol(\CU_{N-d}^2)=\int_{(U,V)\in\CU_{N-d}^2}1\,d(U,V)=\int_{H\in\CU_N\cap\mathcal{T}}\frac{Vol(\xi^{-1}(H))}{NJ(\xi)}\,dH=Vol(\CU_N\cap\mathcal{T})Vol(\CU_{N-2d}).
\]
The value of $Vol(\CU_N\cap\mathcal{T})$ is thus as claimed in Proposition \ref{prop:T}.
\subsubsection{Some notations}
Given a matrix of the form
\begin{equation}\label{eq:H}
H=\begin{pmatrix}0&\sigma&0\\\alpha&C_1&C_2\\0&C_3&C_4\end{pmatrix},
\end{equation}
($\alpha$ and $\sigma$ are $d\times d$ diagonal matrices with real positive ordered entries) we denote by $\tH$ the associated matrix
\[
\tH=\begin{pmatrix}\alpha&C_1&C_2\\0&\sigma&0\\0&U_0^*C_3&U_0^*C_4\end{pmatrix},
\]
where $U_0$ is some unitary matrix which minimizes the distance from $C_4$ to $\CU_{N-2d}$. Note that
\[
\tH=\begin{pmatrix}0&I&0\\I&0&0\\0&0&U_0^*\end{pmatrix}H,
\]
and hence
\[
d(H,\CU_N)=d(\tH,\CU_N).
\]
We also let
\[
T_1(H)=\|\alpha-I_d\|^2+\|\sigma-I_d\|^2+\|C_4-U_0\|^2+\frac{\|C_1\|^2+\|C_2\|^2+\|C_3\|^2}{2},
\]
\[
T_2(H)=\|\alpha-I_d\|^2+\|\sigma-I_d\|^2+\|C_4-U_0\|^2+\|C_1\|^2+\frac{\|C_2\|^2+\|C_3\|^2}{2}=T_1(H)+\frac{\|C_1\|^2}{2}.
\]
Note that
\begin{equation}\label{eq:t1t2d}
T_2(H)\geq T_1(H)\geq\frac{\|\tH-I_N\|^2}{2}\geq\frac{d(\tH,\CU_N)^2}{2}=\frac{d(H,\CU_N)^2}{2}
\end{equation}
\subsubsection{Approximate distance to $\CU_N$ and $\CU_N\cap\mathcal{T}$}
In this section we prove that for small values,
\[
d(H,\CU_N)\approx T_1(H)^{1/2},\quad d(H,\CU_N\cap \mathcal{T})\approx T_2(H)^{1/2}.
\]
More precisely:
\begin{proposition}\label{prop:TtH}
For sufficiently small $\epsilon>0$, if $d(H,\CU_N)\leq \epsilon$ then,
\[
|d(H,\CU_N)-T_1(H)^{1/2}|\leq O(\epsilon^2),
\]
\[
\left|d\left(H,\CU_N \cap\mathcal{T}\right)-T_2(H)^{1/2}\right|\leq O(\epsilon^2).
\]
Here, we are writing $O(\epsilon^2)$ for some function of the form $c(d)\epsilon^2$.
\end{proposition}
Before proving Proposition \ref{prop:TtH} we state the following intermediate result.
\begin{lemma}\label{lem:TtH}
There is an $\epsilon_0>0$ such that $\|\tH-I_N\|\leq \epsilon<\epsilon_0$ implies:
\[
T_1(H)^{1/2}-9\epsilon^2\leq d(H,\CU_N)\leq T_1(H)^{1/2}+9\epsilon^2,
\]
\[
T_2(H)^{1/2}-30\epsilon^2\leq d\left(H,\CU_N \cap\mathcal{T}\right)\leq T_2(H)^{1/2}+30\epsilon^2.
\]
\end{lemma}
\begin{IEEEproof}
We will use the concept of normal coordinates (see for example \cite[p. 14]{Gray04}). Consider the exponential mapping in $\CU_N$, which is given by the matrix exponential
\[
\begin{matrix}
T_I\CU_N=\{R\in\mathcal{M}_N(\C):R+R^*=0\}&\rightarrow&\CU_N\\
R&\mapsto&e^R=I+R+\sum_{k\geq2}\frac{R^k}{k!},
\end{matrix}
\]
which is an isometry from a neighborhood of $0\in T_I\CU_N$ to a neighborhood of $I\in\CU_N$ and defines the normal coordinates. Thus, for sufficiently small $\epsilon_1>0$ there exists $\epsilon_0>0$ such that if $U\in\CU_N$, $\|U-I\|<\epsilon_0$  then there exists a skew--symmetric matrix $R$ such that
\[
U=e^R,\quad \|R\|=d_{\CU_N}(U,I),\quad \|R\|\leq\epsilon_1.
\]

%
Let $R\in\mathcal{M}_N(\C)$ be a skew--Hermitian matrix such that
\[
\|\tH-e^R\|=d(\tH,\CU_N)=\delta\leq\epsilon,\quad  \|R\|=d_{\CU_N}(e^R,I), \quad \|R\|\leq\epsilon_1.
\]
Let $S=\sum_{k\geq2}R^k/k!$. Then, $e^R=I+R+S$ and
\[
\|S\|\leq \sum_{k\geq2}\frac{\|R\|^k}{k!}\leq\|R\|^2.
\]
If we denote $a=\|e^R-I\|=\|R+S\|$ and $b=d_{\CU_N}(e^R,I)=\|R\|$, we have proved that
\[
b-2b^2\leq a\leq b+2b^2.
\]
Assuming that $\epsilon_1<1/3$ (so $b<1/3$) and doing some arithmetic, this implies
\[
a+6a^2\geq b,\quad\text{that is}\quad \|R\|\leq \|e^R-I\|+6\|e^R-I\|^2.
\]
Now,
\[
\|e^R-I\|\leq\|e^R-\tilde{H}\|+\|\tilde{H}-I\|\leq 2\epsilon,
\]
which implies
\[
\|R\|\leq 2\epsilon+24\epsilon^2\leq 3\epsilon.
\]
In particular, $\|S\|\leq 9\epsilon^2$.
We conclude:
\[
d(H,\CU_N)=d(\tH,\CU_N)=\|\tH-e^R\|\geq \|\tH-(I+R)\|-\|S\|\geq \|\tH-(I+R)\|-9\epsilon^2.
\]
We now solve the following elementary minimization problem:
\[
\min_{R:R+R^*=0}\|\tH-(I+R)\|.
\]
Let
\[
R=\begin{pmatrix}R_1&R_2&R_3\\-R_2^*&R_5&R_6\\-R_3^*&-R_6^*&R_9\end{pmatrix},\quad R_1+R_1^*=0,R_5+R_5^*=0,R_9+R_9^*=0.
\]
Then, $\|\tH-(I+R)\|$ is minimized when $R_1=0$, $R_5=0$, $R_9=0$ and
\begin{align*}
R_2=&argmin(\|C_1-R_2\|^2+\|R_2\|^2)\\R_3=&argmin(\|C_2-R_3\|^2+\|R_3\|^2)\\R_6=&argmin(\|U_0^*C_3+R_6^*\|^2+\|R_6\|^2).
\end{align*}
It is easily seen that the solutions to these problems are:
\begin{align*}
R_2=&\frac{C_1}{2}\rightarrow\|C_1-R_2\|^2+\|R_2\|^2=\frac{\|C_1\|^2}{2},\\R_3=&\frac{C_2}{2}\rightarrow\|C_2-R_3\|^2+\|R_3\|^2=\frac{\|C_2\|^2}{2}\\R_6=&-\frac{C_3^*U_0}{2}\rightarrow\|U_0^*C_3+R_6^*\|^2+\|R_6\|^2=\frac{\|C_3\|^2}{2}.
\end{align*}
We have then proved
\[
\min_{R:R+R^*=0}\|\tH-(I+R)\|=T_1(\tH)^{1/2},
\]
and the minimum is reached at
\begin{equation}\label{eq:R}
R=\begin{pmatrix}0&C_1/2&C_2/2\\-C_1^*/2&0&-C_3U_0^*/2\\-C_2^*/2&U_0C_3^*/2&0\end{pmatrix}
\end{equation}
Hence,
\[
d(H,\CU_N)\geq T_1(\tH)^{1/2}-9\epsilon^2,
\]
and the first lower bound claimed in the lemma follows. For the upper bound let $R$ be defined by (\ref{eq:R}) and note that (following a similar reasoning to the one above)
\[
d(H,\CU_N)=d(\tH,\CU_N)\leq \|\tH-e^R\|\leq \|\tilde{H}-(I+R)\|+\sum_{k\geq2}\frac{\|R\|^k}{k!}=T_1(\tH)^{1/2}+\sum_{k\geq2}\frac{\left(\frac{\|C_1\|^2+\|C_2\|^2+\|C_3\|^2}{2}\right)^{k/2}}{k!}.
\]
Now, $\|\tH-I_N\|\leq\epsilon$ in particular implies $\|C_1\|^2+\|C_2\|^2+\|C_3\|^2\leq\epsilon^2$ and then we have
\[
d(H,\CU_N)\leq T_1(\tH)^{1/2}+\sum_{k\geq2}\frac{\left(\frac{\epsilon^2}{2}\right)^{k/2}}{k!}\leq T_1(\tH)^{1/2}+2\epsilon^2,
\]
as wanted. Now, for the second claim of the lemma, the same argument is used but now $R$ is such that $e^R$ minimizes $\|\tH-e^R\|$ and
\[
e^R=\begin{pmatrix}*&*&*\\0&*&*\\ * &*&*\end{pmatrix}.
\]
Now, from the equality
\[
I+R=e^R-S,
\]
and arguing as above we have that
\[
\|R_2\|\leq \|S\|\leq 9\epsilon^2,\text{ which implies }\|R-\tilde{R}\|=\sqrt{2}\|R_2\|\leq20\epsilon^2,
\]
where we denote by $\tilde{R}$ the matrix resulting from letting $R_2=0$. Thus,
\[
|\|\tH-e^R\|-\|\tH-(I+\tilde{R})\||\leq| \|\tH-I-R\|-\|\tH-I-\tilde{R}\||+\|S\|\leq\|\tilde{R}-R\|+9\epsilon^2\leq30\epsilon^2.
\]
We have then proved
\[
\left|d\left(H,\CU_N \cap\mathcal{T}\right)-\min_{R:R+R^*=0,R_2=0}\|\tH-(I+R)\|\right|\leq 30\epsilon^2,
\]
and as before we can easily see that the minimum is reached when $R_1=0$, $R_2=0$, $R_5=0$, $R_9=0$, $R_3=C_2/2$ and $R_6=C_3^*U_0/2$ which proves that
\[
\min_{R:R+R^*=0,R_2=0}\|\tH-(I+R)\|=T_2(H).
\]
This finishes the proof of the lemma.
\end{IEEEproof}
\emph{Proof of Proposition \ref{prop:TtH}}

Let $E$ be a matrix such that $\|E\|\leq\epsilon<1$ and $H=U+E$ for some unitary matrix $U$. Then,
\[
\|HH^*-I\|=\|UU^*+UE^*+EU^*+EE^*-I\|=\|UE^*+EU^*+EE^*\|\leq2\epsilon+\epsilon^2\leq 3\epsilon.
\]
On the other hand,
\[
HH^*-I=\begin{pmatrix}\sigma^2-I&\sigma C_1^*& \sigma C_3^*\\C_1\sigma&X&X\\C_3\sigma&X&C_3C_3^*+C_4C_4^*-I\end{pmatrix},
\]
where the entries $X$ are terms which we do not need to compute. In particular, we have $\|C_1\sigma\|\leq 3\epsilon$ and
\begin{equation}\label{eq:xx1}
\|\sigma^2-I\|\leq 3\epsilon,
\end{equation}
which implies $\|\sigma^{-2}\|=\|\sigma^{-2}-I+I\|\leq \sqrt{d}+4\epsilon$ and hence
\[
\|C_1\|=\|C_1\sigma\sigma^{-1}\|\leq \|C_1\sigma\|\|\sigma^{-1}\|\leq3\epsilon\sqrt{\sqrt{d}+3\epsilon}\leq 4\sqrt{d}\epsilon.
\]
A similar argument works for $C_3$ as well, and using a symmetric argument for $H^*H$ we get the same bound for $C_2$ and an equivalent bound for $\alpha$ to that of (\ref{eq:xx1}). Summarizing these bounds, we have:
\begin{equation}\label{eq:xx2}
\|C_1\|^2+\|C_2\|^2+\|C_3\|^2\leq 48d\epsilon^2
\end{equation}
Moreover, we also have
\[
\|C_4C_4^*-I\|\leq\|C_3C_3^*\|+\|C_3C_3^*+C_4C_4^*-I\|\leq 16d\epsilon^2+4\epsilon\leq20 d\epsilon,
\]
which implies
\[
\sum_{j=0}^{N-d}(\beta_j^2-1)^2=\|C_4C_4^*-I\|^2\leq 400d^2\epsilon^2
\]
 where the $\beta_j$ are the singular values of $C_4$. In particular,
\[
\|U_0^*C_4-I_{N-d}\|^2=d(C_4,U_{N-d})^2=\sum_{j=1}^{N-d}(\beta_j-1)^2\leq \sum_{j=1}^{N-d}(\beta_j-1)^2(\beta_j+1)^2=\sum_{j=1}^{N-d}(\beta_j^2-1)^2\leq 400d^2\epsilon^2,
\]
and we conclude that
\begin{equation}\label{eq:xx3}
\|U_0^*C_4-I_{N-d}\|\leq 20d\epsilon.
\end{equation}
Using (\ref{eq:xx1}), (\ref{eq:xx2}) and (\ref{eq:xx3}) above we get:
\[
\|\tH-I_N\|^2=\|\sigma-I_d\|^2+\|\alpha-I_d\|^2+\|C_1\|^2+\|C_2\|^2+\|C_3\|^2+ \|U_0^*C_4-I_{N-d}\|^2\leq c(d)^2\epsilon^2,
\]
where $c(d)$ depends only on $d$. Let $\epsilon$ be small enough for $c(d)\epsilon$ to satisfy the hypotheses of Lemma \ref{lem:TtH}. The Proposition \ref{prop:TtH} follows from applying that lemma.
\subsubsection{How the sets of closeby matrices to $\CU_N$ and $\CU_N\cap\mathcal{T}$ compare}
Our main result in this section is the following.
\begin{proposition}\label{prop:applyfubini}
Let $\alpha>1$. For sufficiently small $\epsilon>0$, we have:
\[
2^{d^2}Vol\left(H\in\mathcal{T}:d(H,\CU_N\cap \mathcal{T})\leq\frac{\epsilon}{\alpha }\right)\leq
\]
\[
Vol(H\in\mathcal{T}:d(H,\CU_N)\leq\epsilon)\leq
\]
\[
2^{d^2}Vol(H\in\mathcal{T}:d(H,\CU_N\cap \mathcal{T})\leq\alpha \epsilon)
\]
\end{proposition}
Before the proof we state two technical lemmas.
\begin{lemma}\label{lem:fixABvolC}
Let $\sigma,\alpha$ be as in (\ref{eq:H}). Then,
\[
Vol\left(C:T_1\begin{pmatrix}0&\sigma&0\\\alpha&C_1&C_2\\ 0& C_3&C_4\end{pmatrix}\leq\epsilon\right)=2^{d^2} Vol\left(C:T_2\begin{pmatrix}0&\sigma&0\\\alpha&C_1&C_2\\ 0& C_3&C_4\end{pmatrix}\leq\epsilon\right).
\]
\end{lemma}
\begin{IEEEproof}
Let
\[
S_i(C)=T_i\begin{pmatrix}0&A\\B&C\end{pmatrix},\quad i=1,2,
\]
where $A=(\sigma\;0)$ and $B^T=(\alpha\;0)$. The claim of the lemma is that
\[
Vol(C:S_1(C)\leq\epsilon)=2^{d^2}Vol(C:S_2(C)\leq\epsilon).
\]
Indeed, consider the mapping
\[
\varphi\begin{pmatrix}C_1&C_2\\ C_3&C_4\end{pmatrix}= \begin{pmatrix}\sqrt{2}C_1&C_2\\ C_3&C_4\end{pmatrix},
\]
which has Jacobian equal to $\sqrt{2}^{2d^2}=2^{d^2}$. The change of variables theorem yields:
\[
2^{d^2}Vol(C:S_1(\varphi(C))\leq\epsilon)=Vol(\varphi(C):S_1(\varphi(C))\leq\epsilon) =Vol(C:S_1(C)\leq\epsilon).
\]
The lemma follows from the fact that $S_1(\varphi(C))=S_2(C)$.
\end{IEEEproof}
\begin{lemma}\label{lem:fixABvolC2}
Let $\alpha>1$ and let $A,B$ be complex matrices of respective sizes $d\times(N-d)$ and $(N-d)\times d$. Then, for sufficiently small $\epsilon>0$ we have
\[
2^{d^2} Vol\left(C:d\left(\begin{pmatrix}0&A\\B&C\end{pmatrix},\CU_N\cap\mathcal{T}\right)\leq\frac{\epsilon}{\alpha}\right)\leq
\]
\[
Vol\left(C:d\left(\begin{pmatrix}0&A\\B&C\end{pmatrix},\CU_N\right)\leq\epsilon\right)\leq
\]
\[
2^{d^2} Vol\left(C:d\left(\begin{pmatrix}0&A\\B&C\end{pmatrix},\CU_N\cap\mathcal{T}\right)\leq\alpha\epsilon\right).
\]
\end{lemma}
\begin{IEEEproof}
Let $U_A,V_A,U_B,V_B$ be such that
\[
A=U_A(\sigma\;0)V_A^*,\quad B=U_B\binom{\alpha}{0}V_B^*
\]
are singular value decompositions of $A$ and $B$ respectively. Then,
\[
Vol\left(C:d\left(\begin{pmatrix}0&A\\B&C\end{pmatrix},\CU_N\right)\leq\epsilon\right)= Vol\left(C:d\left(\begin{pmatrix}U_A^*&0\\0&U_B^*\end{pmatrix} \begin{pmatrix}0&A\\B&C\end{pmatrix}\begin{pmatrix}V_B&0\\0&V_A\end{pmatrix},\CU_N\right)\leq\epsilon\right)=
\]
\[
Vol\left(C:d\left(\begin{pmatrix}0&(\sigma\;0)\\\binom{\alpha}{0}&U_BCV_A^*\end{pmatrix},\CU_N\right)\leq\epsilon\right)= Vol\left(C:d\left(\begin{pmatrix}0&(\sigma\;0)\\\binom{\alpha}{0}&C\end{pmatrix},\CU_N\right)\leq\epsilon\right),
\]
where the last inequality follows from unitary invariance of the volume. Let $H$ be as in (\ref{eq:H}). From Proposition \ref{prop:TtH}, we conclude:
\[
 Vol\left(C:d(H,\CU_N)\leq\epsilon\right)\leq Vol(C:T_1(H)^{1/2}\leq\epsilon+c(d)\epsilon^2)=
\]
\[
Vol(C:T_1(H)\leq(\epsilon+c(d)\epsilon^2)^2)\underset{\text{Lemma \ref{lem:fixABvolC}}}{=} 2^{d^2}Vol(C:T_2(H)\leq(\epsilon+c(d)\epsilon^2)^2).
\]
From (\ref{eq:t1t2d}), for sufficiently small $\epsilon>0$, $T_2(H)\leq(\epsilon+c(d)\epsilon^2)^2$ implies $d(H,\CU_N)$ is as small as wanted. Hence, from Proposition \ref{prop:TtH}, for sufficiently small $\epsilon>0$ we have
\[
Vol(C:T_2(H)\leq(\epsilon+c(d)\epsilon^2)^2)= Vol(C:T_2(H)^{1/2}\leq\epsilon+c(d)\epsilon^2)\leq
\]
\[
Vol\left(C:d\left(H,\CU_N \cap\mathcal{T}\right)\leq\epsilon+2c(d)\epsilon^2 \right).
\]
In particular, for every $\alpha>1$ and for sufficiently small $\epsilon>0$ we have proved that
\[
 Vol\left(C:d(H,\CU_N)\leq\epsilon\right)\leq 2^{d^2} Vol\left(C:d\left(H,\CU_N \cap\mathcal{T}\right)\leq\alpha\epsilon \right).
\]
This proves the upper bound of the lemma. The lower bound is proved with a symmetric argument, using the opposite inequalities of Proposition \ref{prop:TtH}.

\end{IEEEproof}
\emph{Proof of Proposition \ref{prop:applyfubini}}
Let $\alpha>1$. From Fubini's Theorem,
\[
Vol(H\in\mathcal{T}:d(H,\CU_N)\leq\epsilon)=\int_{A\in\mathcal{M}_{d\times(N-d)}(\C),B\in\mathcal{M}_{(N-d)\times d}(\C)} Vol(C:d(H,\CU_N)\leq\epsilon)\,d(A,B).
\]
From Lemma \ref{lem:fixABvolC2}, for sufficiently small $\epsilon>0$ this is at most
\[
\int_{A\in\mathcal{M}_{d\times(N-d)}(\C),B\in\mathcal{M}_{(N-d)\times d}(\C)} 2^{d^2} Vol(C:d(H,\CU_N\cap\mathcal{T})\leq\alpha\epsilon)\,d(A,B).
\]
Again from Fubini's Theorem, this last equals
\[
2^{d^2}Vol(H:d(H,\CU_N\cap\mathcal{T})\leq\alpha\epsilon),
\]
proving the upper bound of the proposition. The lower bound follows from a symmetrical argument.

\subsubsection{Proof of Proposition \ref{prop:comparevolumes}}
Let $\alpha>1$. From Proposition \ref{prop:applyfubini}, we have
\[
\lim_{\epsilon\rightarrow0}\frac{Vol(H\in\mathcal{T}:d(H,\CU_N)\leq\epsilon)}{\epsilon^{N^2}}\leq 2^{d^2}\lim_{\epsilon\rightarrow0}\frac{Vol(H\in\mathcal{T}:d(H,\CU_N\cap\mathcal{T})\leq\alpha\epsilon)}{\epsilon^{N^2}}.
\]
Note that $N^2$ is the (real) codimension of $\CU_N\cap\mathcal{T}$ inside $\mathcal{T}$.
Thus, from Theorem \ref{th:gray},
\[
\lim_{\epsilon\rightarrow0}\frac{Vol(H\in\mathcal{T}:d(H,\CU_N\cap\mathcal{T})\leq\alpha\epsilon)}{\epsilon^{N^2}}=Vol(\CU_N\cap\mathcal{T})\alpha^{N^2}Vol(x\in\R^{N^2}:\|x\|\leq1).
\]
We have thus proved that for every $\alpha>1$ we have
\[
\lim_{\epsilon\rightarrow0}\frac{Vol(H\in\mathcal{T}:d(H,\CU_N)\leq\epsilon)}{\epsilon^{N^2}}\leq 2^{d^2}Vol(\CU_N\cap\mathcal{T})\alpha^{N^2}Vol(x\in\R^{N^2}:\|x\|\leq1).
\]
This implies:
\[
\lim_{\epsilon\rightarrow0}\frac{Vol(H\in\mathcal{T}:d(H,\CU_N)\leq\epsilon)}{\epsilon^{N^2}}\leq 2^{d^2}Vol(\CU_N\cap\mathcal{T})Vol(x\in\R^{N^2}:\|x\|\leq1).
\]
The reverse inequality is proved the same way using the other inequality of Proposition \ref{prop:applyfubini}.

\subsubsection{Integrals of functions of the subset of matrices in $\mathcal{T}$ which are close to $\CU_N$}
We are now close to the proof of Proposition \ref{prop:limits}, but we still need some preparation. We state two lemmas.
\begin{lemma}\label{lem:almostthere}
Let $\psi:\mathcal{T}\rightarrow[0,\infty)$ be a smooth mapping. Then,
\[
\lim_{\epsilon\rightarrow0}\frac{1}{\epsilon^{N^2}}\int_{H\in\mathcal{T}:d(H,\CU_N)\leq \epsilon}\psi(H)\,dH= 2^{d^2}Vol(\CU_N\cap\mathcal{T})Vol(x\in\R^{N^2}:\|x\|\leq1)
\dashint_{U\in\CU_N\cap\mathcal{T}}\psi(U)\,dU
\]
\end{lemma}
\begin{IEEEproof}
For sufficiently small $\epsilon>0$, given $H\in\mathcal{T}$ such that $d(H,\CU_N)<\epsilon$, there is a unique $U\in\mathcal{U}\cap\mathcal{T}$ such that the distance $d(H,\mathcal{U}\cap\mathcal{T})$ is minimized (see for example \cite[p. 32]{Gray04}). Let $\pi(H)$ be such $U$. Moreover, $\pi$ is a smooth mapping. From Theorem \ref{th:coarea} we thus have
\[
\int_{H\in\mathcal{T}:d(H,\CU_N)\leq \epsilon}\psi\,dH=\int_{U\in\CU_N\cap\mathcal{T}}\int_{H\in\mathcal{T}:d(H,\CU_N)\leq \epsilon,\pi(H)=U}NJ\pi(H)\psi(H)\,dH\,dU.
\]
Now, $\psi$ is smooth and hence $\psi(H)=\Psi(U)+O(\epsilon)$. We thus have
\begin{align*}
\int_{H\in\mathcal{T}:d(H,\CU_N)\leq \epsilon}\psi\,dH&=\int_{U\in\CU_N\cap\mathcal{T}}\psi(U)\int_{H\in\mathcal{T}:d(H,\CU_N)\leq \epsilon,\pi(H)=U}NJ\pi(H)\,dH\,dU\\&+O(\epsilon) Vol(H\in\mathcal{T}:d(H,\CU_N)\leq\epsilon).
\end{align*}
The integral inside this last expression is unitary invariant and thus its value is a constant $c_\epsilon$. Moreover, the same argument applied to $\psi\equiv1$ yields
\[
Vol(H\in\mathcal{T}:d(H,\CU_N)\leq \epsilon)=\int_{U\in\CU_N\cap\mathcal{T}}c_\epsilon\,dU.
\]
That is,
\[
c_\epsilon=\frac{Vol(H\in\mathcal{T}:d(H,\CU_N)\leq \epsilon)}{Vol(\CU_N\cap\mathcal{T})}.
\]
We have then proved
\[
\int_{H\in\mathcal{T}:d(H,\CU_N)\leq \epsilon}\psi\,dH= \frac{Vol(H\in\mathcal{T}:d(H,\CU_N)\leq \epsilon)}{Vol(\CU_N\cap\mathcal{T})} \left(\int_{U\in\CU_N\cap\mathcal{T}}\psi(U)\,dU+O(\epsilon)\right)=
\]
\[
Vol(H\in\mathcal{T}:d(H,\CU_N)\leq \epsilon)\left(\dashint_{U\in\CU_N\cap\mathcal{T}}\Psi(U)\,dU+O(\epsilon)\right).
\]
The lemma follows from Proposition \ref{prop:comparevolumes}.
\end{IEEEproof}

\begin{lemma} \label{lem:limits}
Let $\psi$ be a smooth mapping. Then,
\[
\lim_{\epsilon\rightarrow0}\frac{\int_{H\in\mathcal{T}:d(H,\CU_N)\leq \epsilon}\psi(H)\,dH}{Vol(H\in\mathcal{M}_N(\C):d(H,\CU_N)\leq\epsilon)}=
 \frac{2^{d^2}Vol(\CU_N\cap\mathcal{T})}{Vol(\CU_N)}\dashint_{U\in\CU_N\cap\mathcal{T}}\psi(U)\,dU
\]
\end{lemma}
\begin{IEEEproof}
From Theorem \ref{th:gray} and using that the codimension of $\CU_N$ in $\mathcal{M}_N(\C)$ is $N^2$ we know that
\[
Vol(H\in\mathcal{M}_N(\C):d(H,\CU_N)\leq\epsilon)= Vol(\CU_N)\epsilon^{N^2}Vol(x\in\R^{N^2}:\|x\|\leq1)(1+O(\epsilon)),
\]
where $\lim_{\epsilon\rightarrow0}O(\epsilon)=0$. The lemma now follows from Lemma \ref{lem:almostthere}.
\end{IEEEproof}

\subsubsection{proof of Proposition \ref{prop:limits}}
This result is almost inmediate from Lemma \ref{lem:limits} and Proposition \ref{prop:T}. Let $\xi$ be the mapping defined in (\ref{eq:xi}). We have computed the Normal Jacobian of $\xi$ and the volume of the preimage of $\xi$ in Section \ref{sec:proofpropT}. From Theorem \ref{th:coarea},
\[
\int_{(U,V)\in \CU_{N-d}^2}\Psi(\xi(U,V))\,d(U,V)=\int_{H\in\CU_N\cap\mathcal{T}}\Psi(H)\frac{Vol(\xi^{-1}(H))}{NJ\xi}\,dH=Vol(\CU_{N-2d})\int_{H\in\CU_N\cap\mathcal{T}}\Psi(H)\,dH.
\]
Hence, as $\Psi$ does not depend on $C$, and writing $\Psi(H)=\Psi(A,B)$ (note the abuse of notation),
\[
\int_{H\in\CU_N\cap\mathcal{T}}\Psi(H)\,dH=\frac{1}{Vol(\CU_{N-2d})} \int_{(U,V)\in \CU_{N-d}^2}\Psi\left((I_d\;0)V^*,U\binom{I_d}{0}\right)\,d(U,V).
\]
Normalizing we get
\[
\dashint_{H\in\CU_N\cap\mathcal{T}}\Psi(H)\,dH=\dashint_{(U,V)\in \CU_{N-d}^2}\Psi\left((I_d\;0)V^*,U\binom{I_d}{0}\right)\,d(U,V).
\]
Now, generating at random unitary matrices $U,V$ and then taking $(I_d\;0)V^*,U\binom{I_d}{0}$ is the same as generating at random two elements in the Stiefel manifold $\CU_{(N-d)\times d}$. The proposition is proved.
\subsection{Proof of Theorem \ref{th:main3}}
Recall that we have defined $\CH_\epsilon$ in (\ref{eq:Hepsilon}), and we want to compute the limit (\ref{eq:Hepsilon2}):
\[
\lim_{\epsilon\rightarrow0}\frac{Vol(\CHI\cap\CH_\epsilon)Vol(\CS)}{Vol(\CH_\epsilon)}\dashint_{H\in \CHI\cap\CH_\epsilon}\det(\Psi\Psi^*)\,dH= \lim_{\epsilon\rightarrow0}\frac{Vol(\CS)}{Vol(\CH_\epsilon)}\int_{H\in \CHI\cap\CH_\epsilon}\det(\Psi\Psi^*)\,dH.
\]
Now, we use Fubini's theorem to convert the last integral into an iterated integral
\[
\int_{H_{(k_1,l_1)}\in\mathcal{T},d(H_{(k_1,l_1)},\CU_{N})<\epsilon}\cdots\int_{H_{(k_1,l_1)}\in\mathcal{T},d(H_{(k_r,l_r)},\CU_{N})<\epsilon}\det(\Psi\Psi^*)\,dH_{(k_r,l_r)}\cdots\,dH_{(k_1,l_1)},
\]
where $(k_1,l_1),\ldots,(k_r,l_r)$, $r=K(K-1)$ are all the pairs $(k,l)$ with $k\neq l$, ordered with respect to some (irrelevant) criterion. From Proposition \ref{prop:limits}, the last inner integral satisfies:
\[
\int_{H_{(k_1,l_1)}\in\mathcal{T},d(H_{(k_r,l_r)},\CU_{N})<\epsilon}\det(\Psi\Psi^*)\,dH_{(k_r,l_r)}=O(\epsilon^*)+ Vol(H\in\mathcal{M}_{N}(\C):d(H,\CU_{N})\leq\epsilon)\times
\]
\[
 \frac{2^{d^2}Vol(\CU_{N-d})^2}{Vol(\CU_N)Vol(\CU_{N-2d})}\dashint_{(A^*, B) \in \CU_{(N-d)\times d}}\det(\Psi\Psi^*) \,d(A,B),
\]
where $\Psi$ is computed for
\[
H_{(k_r,l_r)}=\begin{pmatrix}0_{d\times d}&A\\B&0_{(N-d)\times(N-d)}\end{pmatrix}.
\]
Here, $O(\epsilon^*)$ is an expression such that
\[
\lim_{\epsilon\rightarrow0}\frac{O(\epsilon^*)}{Vol(H\in\mathcal{M}_{N}(\C):d(H,\CU_{N})\leq\epsilon)}=0.
\]
By repeating the procedure and using Fubini's theorem again to convert the iterated integral into a unique multiple integral, we conclude:
\[
\int_{H\in \CHI\cap\CH_\epsilon}\det(\Psi\Psi^*)\,dH=O(\epsilon^*)+Vol(H\in\mathcal{M}_{N}(\C):d(H,\CU_{N})\leq\epsilon)^{K(K-1)}\times
\]
\[
 \left(\frac{2^{d^2}Vol(\CU_{N-d})^2}{Vol(\CU_N)Vol(\CU_{N-2d})}\right)^{K(K-1)}\dashint_{(A_{kl}^*, B_{kl}) \in \CU_{(N-d)\times d},k\neq l}\det(\Psi\Psi^*) \,d(A_{kl},B_{kl}),
\]
where $\Psi$ is computed for
\[
H_{kl}=\begin{pmatrix}0_{d\times d}&A_{kl}\\B_{kl}&0_{(N-d)\times(N-d)}\end{pmatrix}.
\]
Here, $O(\epsilon^*)$ is an expression such that
\[
\lim_{\epsilon\rightarrow0}\frac{O(\epsilon^*)}{Vol(H\in\mathcal{M}_{N}(\C):d(H,\CU_{N})\leq\epsilon)^{K(K-1)}}=0.
\]
On the other hand, also from Fubini's theorem we have
\[
Vol(\CH_\epsilon)=Vol(H\in\mathcal{M}_{N}(\C):d(H,\CU_{N})\leq\epsilon)^{K(K-1)}.
\]
The claim of the Theorem \ref{th:main3} follows.

\section{Proof of Theorem \ref{th:numsol_singlebeam}}
\label{sec:proof4}
The proof of this theorem is a generatization of the computation in Section \ref{sec:example}. From Theorem \ref{th:main2}, the number of solutions is given by
\begin{equation}\label{eq:numsols_as_expectation}
\#(\pi_1^{-1}(H_0)) = C \, E\left[|\det(\Psi)|^2 \right],
\end{equation}
where $C$ is the constant defined in Theorem \ref{th:main2} and $\Psi$ is a square matrix of size $L=K(K-1)$. The expectation of the square absolute value of the determinant is
\begin{equation}\label{eq:expectation_squared_determinant}
E[|\det(\Psi)|^2]=E\left[\sum_{\sigma\in S_L}\prod_{i=1}^L\Psi_{\sigma(i) i}\sum_{\delta\in S_L}\prod_{i=1}^L\Psi^*_{\delta(i) i}\right]=E\left[\sum_{\substack{\sigma\in S_L\\\delta \in S_L}}\prod_{i=1}^L\Psi_{\sigma(i) i}\Psi^*_{\delta(i) i}\right],
\end{equation}
where $\sigma,\delta \in S_L$ are permutations of the set $(1,\ldots,L)$, and $\Psi_{ij}$ is the $ij$-th entry of the matrix $\Psi$. 
We note that if $\delta\neq\sigma$ then $\prod_{i=1}^L\Psi_{\sigma(i) i}\Psi^*_{\delta(i) i}$ equals the product of a Gaussian random variable times a non-negative quantity and a quantity depending on other Gaussian variables. By the same argument as in Section \ref{ex:example_th2}, we conclude:
\[
E\left[\prod_{i=1}^L\Psi_{\sigma(i) i}\Psi^*_{\delta(i) i}\right]=0,\quad \sigma\neq\delta.
\]
Thus,
\begin{align}\label{eq:expectation_permanent}
E[|\det(\Psi)|^2] &= E\left[\sum_{\sigma\in S_{K(K-1)}}\prod_{i=1}^{K(K-1)}|\Psi_{\sigma(i) i}|^2\right] \stackrel{(1)}{=}\sum_{\sigma\in S_{K(K-1)}}\prod_{i=1}^{K(K-1)}E[|\Psi_{\sigma(i) i}|^2] \\ \notag
&\stackrel{(2)}{=} \left(\prod_{k\neq l}\frac{1}{(N_kM_l-1)}\right)\sum_{\sigma\in S_{K(K-1)}}\prod_{i=1}^{K(K-1)} \mathbf{1}[\Psi_{\sigma(i) i}\neq 0] \stackrel{(3)}{=} \prod_{k\neq l}\frac{1}{(N_kM_l-1)} \per(T). 
\end{align}
A brief explanation of each step follows:
\begin{itemize}
\item[(1)] Independence among different $\Psi_{\sigma(i) i}$ for a given $\sigma$.
\item[(2)] Every non-zero addend in the sum is the product of $K(K-1)$ independent Beta-distributed random variables. In fact, we note that 
\[
|\Psi_{\sigma(i) i}|^2=\frac{|z_1|^2}{|z_1|^2+|z_2|^2+\cdots+|z_{N_kM_l-1}|^2},
\]
where each $z_i$ is a complex Gaussian random variable, whose real and complex parts are $N(0,1)$ variables (i.e. $z_i$ is a $CN(0,2)$ variable). The distribution of the quotient above is then well known: $|\Psi_{\sigma(i) i}|^2\sim \text{Beta}(1,N_kM_l-2)$ is a beta distribution with parameters $1$ and $N_kM_l-2$, and its expected value equals $E[|\Psi_{\sigma(i) i}|^2]=1/(N_kM_l-1)$ where the values of $k$ and $l$ depend uniquely the row $\sigma(i)$. Therefore, $\prod_{i=1}^{K(K-1)}E[|\Psi_{\sigma(i) i}|^2]=\prod_{k\neq l}\frac{1}{(N_kM_l-1)}$. The notation $\mathbf{1}[P]$ denotes the indicator function which equals 1 if the predicate $P$ is true and 0 otherwise.
\item[(3)] The sum can be identified as a Leibniz-like expansion of the permanent of a (0,1)-matrix $T$ which is built by replacing the non-zero elements of $\Psi$ by ones. More specifically, the matrix $T$ will always have $N_k+M_l-2$ ones per row and $K-1$ ones per column.
\end{itemize}
Combining \eqref{eq:numsols_as_expectation} and \eqref{eq:expectation_permanent}, the compact closed-form expression for the number of solutions in \eqref{eq:numsols_singlebeam_closedform} is obtained.

For the second part of the theorem we note that $T$, with the appropriate row and column ordering, is almost exactly equal to the matrix $A$ obtained by setting $m=n=K$, $w_{ij}=\mathbf{1}[i\neq j]$, $c_i=N_i-1$ and $r_i=K-M_i$ in the notations of \cite[Lemma 9]{Barvinok2010}. To obtain matrix $A$ of \cite[Lemma 9]{Barvinok2010} from our matrix $T$ one just adds $K$ rows containing $M_l$ ones each and $K$ columns containing $K$ ones each. A detailed inspection of the matrices shows that $\per(A)=\per(T)\prod_lM_l$ and then \cite[Lemma 9]{Barvinok2010} implies the second claim of the theorem.


\nocite{Gonzalez2013}
\bibliographystyle{IEEEtran}
\bibliography{myreferences}
\end{document}